\documentclass[aps,pra,reprint,superscriptaddress,nofootinbib,floatfix,longbibliography]{revtex4-2}
\usepackage[T1]{fontenc}
\usepackage{amsmath,amssymb,amsthm,bm,mathtools,mathrsfs}
\usepackage{graphicx,booktabs,array,tabularx,microtype}
\usepackage[hidelinks]{hyperref}
\hypersetup{colorlinks=true,allcolors=blue}
\graphicspath{{figures/}}
\newcommand{\PaperTitle}{Wavefront-conditioned photon measurements for exoplanet spectroscopy: Contrast limits and conditional gains}
\hypersetup{pdftitle={\PaperTitle},pdfauthor={Slava G. Turyshev},pdfsubject={Causal optical measurement, exact information bounds, and exoplanet spectroscopy}}
\newcommand{\Tr}{\operatorname{Tr}}
\newcommand{\E}{\mathbb{E}}
\newcommand{\ket}[1]{\lvert#1\rangle}
\newcommand{\bra}[1]{\langle#1\rvert}
\newcommand{\ddet}{d_{\rm det}}
\newcommand{\drel}{d_{\rm rel}}
\newcommand{\Gwall}{G_{\rm wall}}
\newcommand{\Cpred}{\mathscr C_{\rm pred}}
\newcommand{\Cplanet}{\mathscr C_{\rm planet}}
\newcommand{\Q}{\mathscr Q}
\newtheorem{theorem}{Theorem}
\newtheorem{proposition}[theorem]{Proposition}
\newcolumntype{Y}{>{\raggedright\arraybackslash}X}

\begin{document}
\title{Wavefront-Conditioned Photon Measurements for Exoplanet Spectroscopy: \\ Contrast Limits and Conditional Gains}
\author{Slava G. Turyshev}
\affiliation{Jet Propulsion Laboratory, California Institute of Technology,
4800 Oak Grove Drive, Pasadena, California 91109-0899, USA}
\date{\today}

\begin{abstract}
A wavefront-sensor (WFS) prediction can select the photon measurement before detection or condition its likelihood afterward. This matters for reflected-light spectra requiring hundreds of hours, because subtracting residual starlight cannot remove its photon noise. In a two-mode, low-occupation model, isotropically distributed leakage makes planet-mode projection---ideal single-mode injection matched to the planet---the optimal fixed single-photon measurement, including unequal aligned backgrounds and equal Poisson detector-event rates per output. For nonzero leakage, an attainable two-output conditioned measurement is optimal and strictly more informative before differential loss. For an isotropic background, a contrast bound holds for any distribution of recorded leakage directions: relative transmission $0.80$ and relative admitted-time duty $0.90$ limit the Fisher information per scheduled time to $0.81$ of the fixed optimum at equal predictable-leakage and planet contrasts. Haar averaging sharpens this ceiling to $0.78$. The necessary leakage-to-planet ratios for gain are approximately $2.25$  for any direction law and $3.26$ for Haar directions. A concentration-model certificate identifies directional diversity sufficient to beat recalibrated fixed measurements. Independently gated fixed-basis fallback guarantees no loss of continuum-profiled spectral information when rates and costs are known. For the Habitable Worlds Observatory, the candidate regime is predictable, directionally diverse, leakage-dominated residuals, not a general sensitivity improvement.
\end{abstract}
\maketitle

\section{Introduction}\label{sec:intro}
Reflected-light characterization of an Earth analog is an information-limited measurement. The Habitable Worlds Observatory (HWO) is motivated by the detection and spectroscopy of potentially habitable planets, with roughly $10^{-10}$-class stellar suppression and very small planetary photon rates~\cite{astro2020,Stark2024}. Exposure-time studies show why a characterization can require hundreds of hours: planet continuum, diffuse foregrounds, detector events, and time-dependent residual starlight all contribute to the uncertainty of a molecular estimator~\cite{Stark2025ETC,Ruffio2026}. Removing an estimated stellar mean does not remove its photon fluctuations. Nor can post-processing recover spatial coherence that the optical measurement discarded before photon counting.

Wavefront sensing and control already use predictive information to reduce stellar leakage~\cite{Tesch2026,Pogorelyuk2021,Potier2022}. The question here concerns the residual field \emph{after the same upstream controller has acted}. If a strictly prior wavefront-sensor (WFS) record predicts its spatial direction, can that record improve the science measurement by selecting an optical basis before the next photons arrive? The comparison must give the fixed receiver the same photon budget, controlled field, and prior record in its likelihood. Otherwise an apparent receiver advantage could be an advantage in sensing or calibration rather than in measurement.

The candidate residual is observable and predictable but not adequately rejected by the common controller, for example because the controlling actuator has insufficient bandwidth or cannot access the relevant optical subspace~\cite{Pogorelyuk2021,Potier2022}. A downstream analyzer changes the measurement of that residual rather than its incident power. It is useful only if the retained information exceeds the cost of its own attenuation, detector events, and switching dead time; predictability alone is not a sensitivity gain.

The distinction is not between quantum and classical hardware. Spatial-mode measurements can retain information lost by direct imaging~\cite{Tsang2016}, and quantum-optimal coronagraphy has been developed theoretically and demonstrated in laboratory mode-sorting experiments~\cite{Deshler2025,Deshler2026}. Recent analyses include finite stellar diameter, known planet locations, complex pupils, and candidate HWO mode-sorting and multi-plane light-conversion designs~\cite{Xin2026,XinHWO2026,Haffert2026}. The device considered here likewise needs only passive linear optics, a programmable two-mode unitary, and photon counting. A classical optical instrument implementing the same conditioned unitary is the proposed receiver, not an excluded competitor.

WFS-conditioned post-detection inference is also established~\cite{Frazin2013,Rodack2021,XinWFE2024}. More generally, quantum information theory distinguishes side information available before and after a measurement~\cite{Ballester2008,Carmeli2018}. Adaptation from accumulated science photons, as considered by Choi \emph{et al.}~\cite{Choi2026}, is a different timing resource. Here the setting is a known planet and an independent, earlier WFS record, not scene discovery from the same low-flux science stream. A broader discussion of how measurement-level gains propagate to astrophysical inference is given in Ref.~\cite{Turyshev2026Benchmarks}.

The reference measurement is physically recognizable: its informative port is ideal single-mode injection matched to the known planet point-spread function, while the orthogonal output is also counted. For a unitarily invariant, or Haar, ensemble of two-mode leakage directions, this planet projection is the exact global fixed optimum. Haar isotropy describes the relative complex modal amplitudes, not an isotropic sky brightness or an empirical HWO disturbance distribution.

The astronomical interpretation is controlled by photons sharing the planet mode. Planet continuum fluctuates in the same spatial mode as the molecular derivative and cannot be rejected without also losing signal. For an isotropic optical floor, a direction-distribution-independent contrast bound quantifies this restriction: with relative transmission $0.80$ and relative admitted-time duty $0.90$, equal predictable-leakage and planet contrasts imply at most $0.81$ of the optimal fixed receiver's information per scheduled time. Haar averaging tightens the ceiling to $0.779$ but is not needed for the exclusion. Both statements retain the two-mode conditional-Poisson model and matched detector rates. A possible opportunity is instead a localized, leakage-dominated modal pair during degraded suppression or at a difficult working angle; its direction must remain predictable while varying enough that a recalibrated fixed basis cannot recover the same information.

Five results organize the analysis. Theorem~\ref{thm:fixed} gives the global fixed optimum, including known unequal aligned optical backgrounds. Section~\ref{sec:noisy_optimum} proves that the two-output symmetric logarithmic derivative (SLD) measurement attains the conditioned optimum after scalar loss and matched detector charges. Propositions~\ref{prop:ceiling} and~\ref{prop:direction_free} bound that optimum by modal contrast, with and without Haar averaging. Equation~\eqref{eq:diversity_lower} certifies gains against arbitrary fixed measurements recalibrated at the specified cadence in a concentration model. Equation~\eqref{eq:fallback} gives the independent-channel fallback corollary for continuum-profiled spectral information. The conditional count likelihood is specified in Sec.~\ref{sec:model}; Secs.~\ref{sec:discussion} and~\ref{sec:numerics} connect these results to controlled fields and reproducible numerical tests; Sec.~\ref{sec:conclusions} discusses results obtained  and concludes.

\section{Causal measurement and statistical model}\label{sec:model}
\subsection{Common resources and the meaning of causality}\label{sec:causal_resources}
A measurement is causal here in the signal-processing sense: its setting may depend only on information available before the detected photon. Figure~\ref{fig:architecture} shows the common test architecture. A unitary modal mapper feeds a programmable two-mode analyzer, both output ports are counted, and the event record stores the prior WFS data, the latched command, and the relevant transfer-matrix calibration. A predeclared policy selects either a stored fixed basis or a basis calculated from the prior WFS record. The downstream likelihood retains the same record under both policies.

\begin{figure*}[t]
\centering
\includegraphics[width=0.98\textwidth]{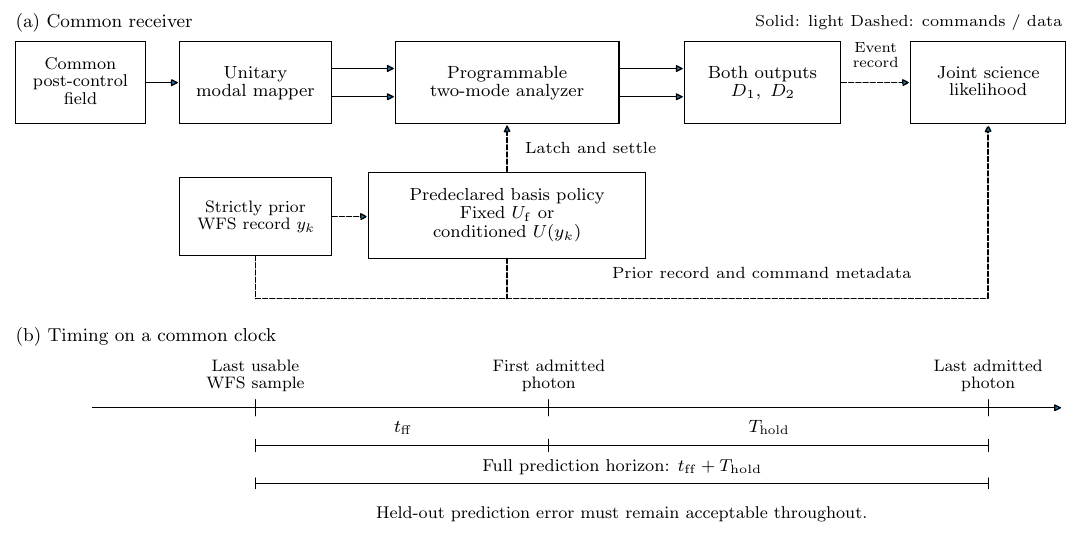}
\caption{Common receiver and causal timing. (a) Both policies share the mapper, programmable analyzer, and two counted outputs. Solid paths carry light; dashed paths carry commands or recorded data. The prior wavefront-sensor (WFS) record selects $U(y_k)$ only under the conditioned policy, but enters both likelihoods with the latched command and calibration identifiers. (b) Admission begins after latching and settling; held-out prediction error must remain acceptable through the last admitted photon. Equal calibrated losses give $\tau=1$; the analysis also permits measured differential attenuation. The diagram is a timing contract, not a device-performance measurement.}
\label{fig:architecture}
\end{figure*}

For interval $k$, let $y_k$ contain time-stamped WFS measurements and upstream-control commands available before the last usable sample time $t_{z,k}^{\rm last}$. Define
\begin{align}
 t_{\rm ff}&=\Delta t_{\rm lat}+\Delta t_{\rm solve}
             +\Delta t_{\rm switch}+\Delta t_{\rm margin},\nonumber\\
 t_{z,k}^{\rm last}+t_{\rm ff}&<t_k^{\rm open}.
 \label{eq:causality}
\end{align}
All times are in seconds ($\mathrm s$). The feed-forward delay $t_{\rm ff}$ includes transport latency, basis calculation, switching and settling, and a timing margin. If the basis is held for $T_{\rm hold}$ after opening, its prediction must remain valid over the entire photon-admission interval, not just at its first photon. In particular, the largest prediction horizon is approximately $t_{\rm ff}+T_{\rm hold}$. Its error is measured on held-out data; an estimate smoothed with photons or WFS samples from the same interval is not a causal command.

The analyzer update interval $T_{\rm upd}$ sets the command cadence; the fixed-basis calibration interval $T_{\rm cal}$ sets how long a stored basis is used before recalibration. All photon budgets below use physical rates per second. A unitary optical matrix $U$ is dimensionless; the fixed and conditioned settings are $U_{\rm f}$ and $U(y_k)$. The two detectors are labeled $D_1$ and $D_2$ in Fig.~\ref{fig:architecture}.

\subsection{Rate operators and conditional counts}
In one wavelength--polarization block, let $P=\ket p\bra p$ be the known, normalized planet mode and let $S=\ket s\bra s$ be the predicted normalized stellar-leakage mode. These dimensionless rank-one projectors describe relative modal amplitudes and phase, not optical power or sky position. The identity $I$ acts on the two selected orthogonal spatial modes. The reduced optical rate operator is
\begin{align}
 R_s(c)&=B_s+cP,\nonumber\\
 B_s&=\beta I+AS+CP
     =\beta(I+rS+\gamma P),\label{eq:reduced}\\
 r&=A/\beta,\qquad \gamma=C/\beta.\nonumber
\end{align}
Here $\beta>0$ is an irreducible optical rate \emph{per mode}, $A\ge0$ is the total predictable leakage rate in the selected pair, and $C\ge0$ is the planet-continuum rate. The additive local spectral coordinate $c$, as well as $A$, $C$, and $\beta$, has units photon $\mathrm{s}^{-1}$. Photon number is dimensionless in the information calculations; all rates can therefore be expressed in $\mathrm{s}^{-1}$. The ratios $r$ and $\gamma$ are dimensionless. A neighborhood of the fiducial point is restricted to positive physical rates. The signed molecular parameter is introduced in Sec.~\ref{sec:spectroscopy}.

A positive-operator-valued measurement (POVM) $\{\Pi_m\}$ specifies dimensionless effects $\Pi_m\succeq0$ with $\sum_m\Pi_m=I$. Here $\succeq0$ means positive semidefinite, $\Tr$ denotes the matrix trace, and a dagger denotes Hermitian conjugation. Low occupation means that the mean photon number per collected optical spatiotemporal mode is much smaller than one. We adopt a conditional Poisson rare-event model with independent output increments: for an admitted duration $t$, the counts obey
\begin{equation}
 n_m\mid s\sim\operatorname{Pois}\!\left[t\{\Tr(\Pi_mR_s(c))+\ddet\}\right].
 \label{eq:poisson}
\end{equation}
The notation $\operatorname{Pois}[\lambda]$ denotes a Poisson distribution with dimensionless mean $\lambda$. The known stationary detector-event rate $\ddet\ge0$, in events $\mathrm{s}^{-1}$, is charged to each \emph{physical detector output}, not in proportion to $\Tr\Pi_m$. The fixed receiver retains $s$ in its likelihood but cannot select $\Pi_m$ from $s$ within the fixed interval. The conditioned receiver can select $\Pi_m(s)$. The WFS record is assumed ancillary for the local planetary parameter: its probability law does not itself depend on $c$. A common parameter-sensitive WFS score would have to be added to both branches and would change a gain ratio.

Where an exact ensemble average is required, $s$ is taken to be Haar distributed in $\mathbb C^2$, equivalently uniform in its Bloch-sphere representation. The dimensionless overlap $u=|\langle p|s\rangle|^2$ is then uniform on $[0,1]$. This is a solvable model of modal diversity, not a distribution inferred from an HWO disturbance spectrum. Proposition~\ref{prop:direction_free} does not require this direction law. The receiver class consists of measurements of individual photons, with no quantum memory between detections and any finite number of outcomes. Collective measurements, a phase-referenced local oscillator, and persistent unresolved latent intensity are outside Eq.~\eqref{eq:poisson}; the latter is discussed in Sec.~\ref{sec:latent}.

\subsection{Information per admitted and scheduled time}
At $c=0$, the local Fisher information per unit \emph{admitted} time is
\begin{equation}
 J(s;\Pi)=\sum_m\frac{[\Tr(P\Pi_m)]^2}{\Tr(B_s\Pi_m)+\ddet}.
 \label{eq:fi}
\end{equation}
Because $c$ is a rate, $J$ has units seconds, while $tJ$ has units $\mathrm{s}^2$. For a dimensionless amplitude with derivative $gP$, where $g$ is a rate, its information rate is $g^2J$ and has units $\mathrm{s}^{-1}$. We use $J_{\rm f}$ and $J_{\rm c}$ for fixed and conditioned information rates, respectively; their units depend on the parameter being estimated. The symbol $\E_s$ denotes expectation over the declared direction law, or the equivalent weighted average along a recorded direction history.

The fair pre-/postmeasurement comparison is
\begin{equation}
 J_{\rm pre}=\E_s\!\left[\sup_{\Pi(s)}J(s;\Pi(s))\right],\qquad
 J_{\rm post}=\sup_\Pi\E_s[J(s;\Pi)].
 \label{eq:prepost}
\end{equation}
The ordering $J_{\rm pre}\ge J_{\rm post}$ is immediate; strictness and the optimal fixed measurement are not. Their calculation is the purpose of Sec.~\ref{sec:measurement}.

Let $d_{\rm c}$ and $d_{\rm f}$ be the admitted science-time fractions per scheduled wall time. The relative duty and local wall-information ratio are
\begin{equation}
 \drel=\frac{d_{\rm c}}{d_{\rm f}},\qquad
 \Gwall=\frac{d_{\rm c}J_{\rm c}}{d_{\rm f}J_{\rm f}}.
 \label{eq:wallgain}
\end{equation}
The ratio $\drel$ can exceed one, although each absolute duty is at most one. At fixed local significance, with the same identifiable model and no non-scaling prior, a dimensionless information ratio $G$ multiplies the fixed-time signal-to-noise ratio (SNR) by $\sqrt G$ and gives an asymptotic exposure reduction $1-G^{-1}$. It is not itself a finite-count completeness calculation. The duties in scalar formulas below are assumed independent of the instantaneous direction; state-dependent rejection requires averaging the admitted information with its actual selection probabilities.

\section{Exact fixed optimum and conditioned information}\label{sec:measurement}
\subsection{A fixed optimum that includes aligned anisotropy}
The fixed optimum can be proved in a slightly more general model than Eq.~\eqref{eq:reduced}. Write
\begin{equation}
 B_s=b_{\parallel}P+b_{\perp}P_{\perp}+AS,\qquad
 P_{\perp}=I-P,
 \label{eq:aligned}
\end{equation}
where $b_{\parallel},b_{\perp}>0$ are known optical floors in $\mathrm{s}^{-1}$. The planet continuum is included in $b_{\parallel}$; the isotropic case has $b_{\parallel}=\beta+C$ and $b_{\perp}=\beta$.

\begin{theorem}[Optimal fixed measurement]\label{thm:fixed}
For the Haar ensemble in Eq.~\eqref{eq:aligned}, with derivative $P$ and equal known detector rates $\ddet$ per physical output, the globally optimal fixed memoryless single-photon POVM is $\{P,P_{\perp}\}$. For arbitrary finite outcome number,
\begin{equation}
 J_{\rm f}^{\star}
 =\frac1A\ln\!\left(1+\frac{A}{b_{\parallel}+\ddet}\right),
 \label{eq:fixed_general}
\end{equation}
with continuous value $(b_{\parallel}+\ddet)^{-1}$ at $A=0$.
\end{theorem}
\begin{proof}
For one effect $\Pi$, let its eigenvalues be $0\le\lambda_-\le\lambda_+\le1$, its trace be $v=\lambda_-+\lambda_+$, and its planet overlap be $a=\Tr(P\Pi)\in[\lambda_-,\lambda_+]$. Haar invariance gives
\begin{equation}
 \langle s|\Pi|s\rangle=\lambda_-+(\lambda_+-\lambda_-)z,\qquad
 z\sim\operatorname{Uniform}(0,1).
 \label{eq:haar_effect}
\end{equation}
The mean information contribution divided by $a>0$ is $\int_0^1 a/D\,dz$, where
\begin{align}
 D={}&b_{\perp}v+(b_{\parallel}-b_{\perp})a+\ddet\nonumber\\
 &+A[\lambda_-+(\lambda_+-\lambda_-)z].
 \label{eq:effect_denominator}
\end{align}
For fixed eigenvalues and $z$, $a/D$ increases with $a$, since
\begin{equation}
 \frac{\partial(a/D)}{\partial a}
 =\frac{b_{\perp}v+\ddet+A\langle s|\Pi|s\rangle}{D^2}>0.
 \label{eq:effect_derivative}
\end{equation}
This remains true even when $b_{\parallel}<b_{\perp}$. Set $a=\lambda_+$. Then
\begin{equation}
 D=\lambda_+(b_{\parallel}+Az)
   +\lambda_-[b_{\perp}+A(1-z)]+\ddet.
 \label{eq:aligned_denominator}
\end{equation}
The integrand $\lambda_+/D$ decreases with $\lambda_-$, so it is bounded above by its value at $\lambda_-=0$. That value increases with $\lambda_+$, its derivative being
\begin{equation}
 \frac{\ddet}{[\ddet+\lambda_+(b_{\parallel}+Az)]^2}\ge0.
 \label{eq:last_effect_derivative}
\end{equation}
Consequently the contribution of any effect is at most
$a\int_0^1(b_{\parallel}+\ddet+Az)^{-1}dz$.
Completeness gives $\sum_m a_m=1$. Summing establishes Eq.~\eqref{eq:fixed_general}, and the planet projection attains it.
\end{proof}

This effectwise proof avoids refining a physical detector port into several independently noisy ports. It covers unequal aligned floors without a numerical POVM search. Its physical content is simple: the fixed optimum sends the entire local planet derivative to the planet-matched output. The known orthogonal floor does not enter that output's rate. It still affects a conditioned analyzer that mixes the two modes. An unknown floor, a nonaligned covariance, or a different nuisance derivative is a different optimization problem.

\subsection{An explicit attainable conditioned measurement}
For a positive two-mode rate matrix, the SLD for the additive rate $c$ is a Hermitian matrix $L_s$ with units seconds. It solves
\begin{equation}
 B_sL_s+L_sB_s=2P.
 \label{eq:sld}
\end{equation}
Its eigenbasis attains the scalar quantum Fisher-information rate for derivative $P$~\cite{BraunsteinCaves1994}. Here it is useful to give the solution explicitly. In the planet basis, put
\begin{equation}
 B_s=\begin{pmatrix}a_s&z_s\\z_s^*&v_s\end{pmatrix},\quad
 \mathcal T_s=a_s+v_s,\quad
 \Delta_s=a_sv_s-|z_s|^2.
 \label{eq:matrix_invariants}
\end{equation}
The entries $a_s,v_s,z_s$ and trace $T_s$ have units $\mathrm{s}^{-1}$, and the determinant $\Delta_s$ has units $\mathrm{s}^{-2}$. Direct substitution gives
\begin{align}
 L_s&=\frac1{\Delta_s\mathcal T_s}
 \begin{pmatrix}v_s^2+\Delta_s&-v_sz_s\\-v_sz_s^*&|z_s|^2\end{pmatrix},
 \label{eq:sld_explicit}\\
 J_s^Q&=\Tr(PL_s)=\frac{v_s^2+\Delta_s}{\Delta_s\mathcal T_s}.
 \label{eq:qfi}
\end{align}
For any SLD eigenvector $\ket{e_m}$, the measured background is $b_m=\langle e_m|B_s|e_m\rangle$ and the rate slope is $p_m=|\langle e_m|p\rangle|^2$. If its SLD eigenvalue is $\ell_m$, Eq.~\eqref{eq:sld} gives $p_m=\ell_m b_m$. Hence $\sum_m p_m^2/b_m=\Tr(B_sL_s^2)=\Tr(PL_s)$, which verifies attainability in the two counted outputs without additional measurement resources.

The planet projection has information $1/a_s$. Their difference is
\begin{equation}
 J_s^Q-\frac1{a_s}
 =\frac{v_s|z_s|^2}{a_s\Delta_s\mathcal T_s}>0
 \quad\text{when }z_s\ne0.
 \label{eq:strict}
\end{equation}
For Eq.~\eqref{eq:aligned}, $|z_s|^2=A^2u(1-u)$. Thus every $A>0$ gives strict ensemble separation before differential loss, since $u=0,1$ have Haar measure zero. The same statement holds with an equal detector floor and equal branch transmission by replacing $B_s$ with $B_s+\ddet I$ for the two-output projective measurements.

Equation~\eqref{eq:strict} identifies the resource: direction-dependent \emph{spatial} coherence between the planet basis modes. No optical phase reference to the star is required. A leakage-nulling basis minimizes a port's stellar rate but need not retain the most informative planet score. The SLD instead balances slope against total noise in both denominators. It does not reduce incident leakage, duplicate photons, or obtain a gain by discarding a bright port. A WFS record used only afterward cannot change the coherence that a fixed detection basis resolved.

\subsection{Isotropic ensemble average and implementation costs}\label{sec:noisy_optimum}
For Eq.~\eqref{eq:reduced}, define dimensionless invariants
\begin{equation}
 v=1+r(1-u),\qquad \Delta=1+r+\gamma v,\qquad
 \mathcal T=2+r+\gamma.
 \label{eq:dimensionless_invariants}
\end{equation}
Equations~\eqref{eq:fixed_general} and~\eqref{eq:qfi} give, without detector events,
\begin{align}
 J_{\rm f}^{\star}&=\frac{\ln[1+r/(1+\gamma)]}{\beta r},
 \label{eq:fixed_isotropic}\\
 J_{\rm c}&=\frac{\mathcal C(r,\gamma)}{\beta},\qquad
 \mathcal C(r,\gamma)=\frac{1+\mathcal I(r,\gamma)}{2+r+\gamma},\nonumber\\
 \mathcal I(r,\gamma)&=\frac1r\int_1^{1+r}
            \frac{v^2}{1+r+\gamma v}\,dv.
 \label{eq:causal_average}
\end{align}
For $h=1+r$ and $\gamma>0$, polynomial division yields
\begin{equation}
 \mathcal I=\frac{r+2}{2\gamma}-\frac{h}{\gamma^2}
       +\frac{h^2}{r\gamma^3}\ln\!\left[\frac{h(1+\gamma)}{h+\gamma}\right].
 \label{eq:closed_integral}
\end{equation}
Cancellation makes this expression unsuitable at small $\gamma$; the bounded orientation integral is numerically stable. The exact continuous limits are
\begin{align}
 \mathcal C(r,0)&=\frac{r^2+6r+6}{3(r+1)(r+2)},\nonumber\\
 \mathcal C(0,\gamma)&=\frac1{1+\gamma}.
 \label{eq:limits}
\end{align}
At vanishing leakage the intrinsic separation vanishes. At fixed $\gamma$ and very large $r$, a direction-conditioned output can retain information in a weakly contaminated mode; the fixed planet port encounters the full range of leakage overlaps. Increasing planet continuum, however, adds fluctuations directly in the same mode as the derivative and consumes this advantage.

Let $0<\tau\le1$ be conditioned-to-fixed optical transmission after common telescope and spectral optics, and let $\kappa=\ddet/\beta$. Both $\tau$ and $\kappa$ are dimensionless. Optical loss attenuates both signal slope and optical background, whereas detector events are added afterward. The matched-detector fixed optimum and the two-output conditioned SLD give
\begin{align}
 J_{\rm f}^{\star}&=\frac1{\beta r}
        \ln\!\left(1+\frac r{1+\gamma+\kappa}\right),\nonumber\\
 r'&=\frac{\tau r}{\tau+\kappa},\qquad
 \gamma'=\frac{\tau\gamma}{\tau+\kappa},\nonumber\\
 J_{\rm c}&=\frac{\tau^2}{\beta(\tau+\kappa)}
        \mathcal C(r',\gamma'),\label{eq:detector_rates}\\
 \Gwall&=\drel\frac{\tau^2r}{\tau+\kappa}
 \frac{\mathcal C(r',\gamma')}{\ln[1+r/(1+\gamma+\kappa)]}.
 \label{eq:exact_gain}
\end{align}
The two-output SLD is also globally optimal at each direction among all finite-outcome memoryless single-photon POVMs under the stated loss and detector model. To see why the per-output detector charge does not invalidate that bound in two dimensions, note that Eq.~\eqref{eq:fi} is convex in the effects. At fixed finite outcome number its maximum can therefore be attained at an extreme POVM. Every nontrivial extreme two-dimensional POVM has rank-one effects: a full-rank effect and any other nonzero effect admit a small opposite transfer of the latter between the two, producing a nontrivial convex decomposition. A rank-one effect obeys $\Tr\Pi_m\le1$, so its physical detector charge is at least $\ddet\Tr\Pi_m$. Its information is consequently no greater than that for the rate operator $B_s+\ddet I$, which is bounded by its SLD information. The remaining one-output measurement has information $(\Tr B_s+\ddet)^{-1}$ and is no better than the planet projection. Convexity extends the bound to every finite-outcome POVM, and the two-output SLD attains it. The same reasoning applies after relative optical attenuation. Thus $J_{\rm c}$ in Eq.~\eqref{eq:detector_rates} is the optimum over all direction-conditioned measurements in this class, not only over projective analyzers. This argument is specific to two dimensions and the equal, known, independent detector-event model; unequal rates, transients, and mode-dependent transmission require the implemented likelihood.

In the same-hardware experiment of Fig.~\ref{fig:architecture}, equal loss at the two settings gives $\tau=1$. Values below one represent measured setting-dependent loss or a comparison with a more transmissive fixed implementation. The scalar model requires a direction-independent relative attenuation; it is not a substitute for a transfer matrix when loss varies by mode or command.

For an assigned target $G_{\rm tar}$, the required relative duty is
\begin{equation}
 d_{\rm rel}^{\rm req}
 =\frac{G_{\rm tar}}{J_{\rm c}/J_{\rm f}^{\star}},\qquad
 \drel\ge d_{\rm rel}^{\rm req}.
 \label{eq:required_duty}
\end{equation}
Only if $d_{\rm c}\le d_{\rm f}$ may one additionally require $d_{\rm rel}^{\rm req}\le1$. For $G_{\rm tar}=1.25$, a 20\% local exposure saving, $\gamma=0.5$, $\tau=0.8$, and $\kappa=0.1$, the leakage threshold is $r=11.359$ at relative duty one and $r=14.582$ at $d_{\rm c}=0.9$, $d_{\rm f}=1$. The stress point $r=10$ requires $d_{\rm rel}^{\rm req}=1.0499$; it misses this target when the conditioned branch has no duty advantage, but is not excluded at every possible ratio of absolute duties.

\begin{table}[t]
\caption{Successive implementation charges in assigned two-mode examples. $G$ is the appropriate information ratio, including the listed duty. The saving is $100(1-G^{-1})$. The final row is a different, more favorable operating point, not an additional step in the first four rows.}
\label{tab:penalties}
\centering\small
\begin{tabular}{cccccrr}
\toprule
$r$&$\gamma$&$\tau$&$\kappa$&$d_{\rm c}/d_{\rm f}$&$G$&Saving (\%)\\
\midrule
10&0&0.80&0&1&1.3985&28.5\\
10&0.5&0.80&0&1&1.2399&19.3\\
10&0.5&0.80&0.1&1&1.1906&16.0\\
10&0.5&0.80&0.1&0.90&1.0715&6.7\\
15&0.5&0.90&0.1&0.90&1.4352&30.3\\
\bottomrule
\end{tabular}
\end{table}

Table~\ref{tab:penalties} separates continuum, detector, and duty penalties. Figure~\ref{fig:implementation}(a) shows how transmission changes the required duty. Fixed recalibration is not free: if a refresh occupies $t_{\rm ref}$ of each $T_{\rm cal}$ without overlapping another interruption, then $d_{\rm f}=1-t_{\rm ref}/T_{\rm cal}$. At $d_{\rm c}=0.90$, illustrative fixed losses of 3\% and 5\% give relative duties $0.9278$ and $0.9474$, lowering the 20\%-saving root to $13.608$ and $12.962$. A change in calibration cadence also changes the fixed receiver's information, not just its duty; that second effect is treated in Sec.~\ref{sec:temporal}.

\begin{figure*}[t]
\centering
\begin{minipage}{0.485\textwidth}\centering\includegraphics[width=\linewidth]{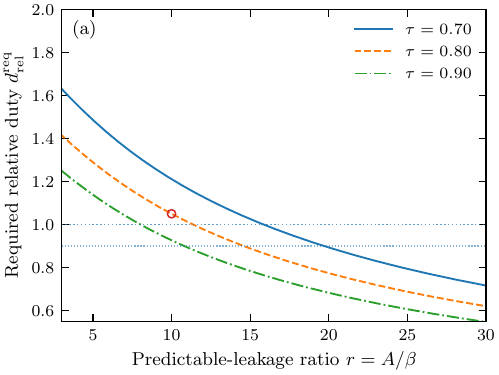}\end{minipage}\hfill
\begin{minipage}{0.485\textwidth}\centering\includegraphics[width=\linewidth]{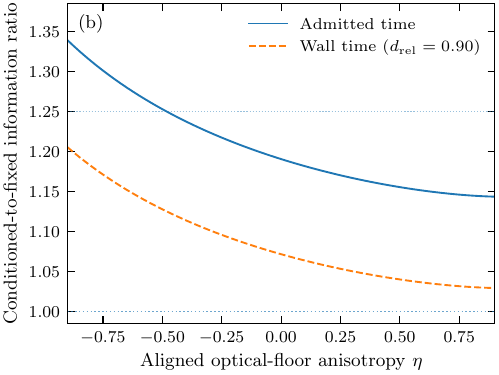}\end{minipage}
\caption{Implementation and background geometry. (a) Required relative duty for $G_{\rm tar}=1.25$ at $\gamma=0.5$, $\kappa=0.1$, and three relative transmissions. The horizontal guide at one is a restriction only when $d_{\rm c}\le d_{\rm f}$. The open marker shows the $r=10$, $\tau=0.8$ stress point. (b) Exact aligned-floor comparison at $(r,\gamma,\tau,\kappa)=(10,0.5,0.8,0.1)$ for $b_{\parallel}=\beta(1+\eta)+C$ and $b_{\perp}=\beta(1-\eta)$. The fixed comparator is given by Theorem~\ref{thm:fixed}; the two curves distinguish admitted and wall time. These are scalar-model sensitivities, not an HWO field distribution.}
\label{fig:implementation}
\end{figure*}

The aligned-floor extension also clarifies a useful stress test. Set
\begin{equation}
 b_{\parallel}=\beta(1+\eta)+C,\qquad
 b_{\perp}=\beta(1-\eta),\qquad |\eta|<1.
 \label{eq:anisotropy}
\end{equation}
At the stress point of Table~\ref{tab:penalties}, the wall ratio changes from $1.2056$ at $\eta=-0.9$ to $1.0292$ at $\eta=0.9$, with $1.0715$ at zero. A brighter planet-mode floor favors the fixed comparison in this assigned family, as shown in Fig.~\ref{fig:implementation}(b). Theorem~\ref{thm:fixed} makes this an exact comparison for aligned anisotropy. It does not extend to a floor whose principal axes are misaligned with $P$, and its percentage change cannot be applied as a universal correction to a different operating point.

\section{Contrast bounds and their photon-budget interpretation}\label{sec:budget}
\subsection{A pointwise optical bound and its Haar average}
A favorable measurement ratio at independently assigned $r$ and $\gamma$ need not describe a physical planet. Both planet continuum and predictable leakage refer to the same stellar channel rate $F_\star$. Define
\begin{align}
 A&=F_\star\Cpred,\qquad C=F_\star\Cplanet,\nonumber\\
 x&=\frac AC=\frac{\Cpred}{\Cplanet}=\frac r\gamma.
 \label{eq:contrast_ray}
\end{align}
The contrasts $\Cpred$ and $\Cplanet$, and their ratio $x$, are dimensionless: they are rates in the \emph{selected modal pair}, normalized to the same stellar channel rate $F_\star$ in $\mathrm{s}^{-1}$. They are not automatically a spatially averaged dark-hole contrast. Changing the isotropic optical floor at fixed $A$ and $C$ moves the point along a fixed $x$, while the detector rate must remain a separate physical quantity.

\begin{proposition}[Haar contrast bounds]\label{prop:ceiling}
For Eq.~\eqref{eq:reduced}, Haar $s$, $A,C>0$, $\beta,\ddet\ge0$, $0<\tau\le1$, and direction-independent positive duties, let $\Gwall$ be the optimal conditioned-to-fixed information ratio, attained by the two-output SLD of Sec.~\ref{sec:noisy_optimum}. Then
\begin{align}
 \drel\tau^2\ \le\ \Gwall
 &\le\drel\tau\,\Phi\!\left(\frac A{C+\beta+\ddet}\right)\nonumber\\
 &\le\drel\tau\,\Phi(x),\label{eq:ceiling}\\
 \Phi(x)&=\frac{x(x+2)}{2(x+1)\ln(1+x)}
        =\frac{\sinh[\ln(1+x)]}{\ln(1+x)}.
 \label{eq:phi}
\end{align}
Every direction-conditioned finite-outcome memoryless single-photon POVM obeys the upper bounds in Eq.~\eqref{eq:ceiling}. The lower bound is guaranteed for the optimal conditioned receiver, not for an arbitrary suboptimal measurement. The $A=0$ limits are continuous. These statements are independent of the isotropic-background partition but retain the Haar ensemble, scalar loss, and equal independent per-output detector-event assumptions.
\end{proposition}
\begin{proof}
For derivative $P$, introduce the variational information functional
\begin{equation}
 \Q(B;P)=\sup_{L=L^\dagger}
       \{2\Tr(PL)-\Tr(BL^2)\}.
 \label{eq:q_variation}
\end{equation}
The functional $\Q$ has units seconds. Its stationary equation is Eq.~\eqref{eq:sld}, and its optimum is $\Tr(PL)$. Since $L^2\succeq0$, adding a positive rate operator cannot increase $\Q$. Put $D=C+\beta+\ddet$, the direction-independent planet-port rate in $\mathrm{s}^{-1}$. Factoring the conditioned transmission gives
\begin{equation}
 AS+CP+(\beta+\ddet/\tau)I\succeq AS+DP.
 \label{eq:operator_comparison}
\end{equation}
The difference has nonnegative eigenvalues $\ddet(1/\tau-1)$ in the planet mode and $\beta+\ddet/\tau$ in its complement. At overlap $u$, Eq.~\eqref{eq:qfi} gives
\begin{equation}
 \Q(AS+DP;P)=\frac{D+A(1-u)}{D(A+D)}.
 \label{eq:zero_floor_pointwise}
\end{equation}
Singular endpoints are defined by continuity on the support of $P$. Averaging uniform $u$ therefore gives $(A+2D)/[2D(A+D)]$. By Sec.~\ref{sec:noisy_optimum}, the largest possible conditioned information is
$\tau\E_s\Q[AS+CP+(\beta+\ddet/\tau)I;P]$,
whereas Eq.~\eqref{eq:fixed_general} gives the fixed information $\ln(1+A/D)/A$. Division proves the first upper bound for the optimum and hence for every conditioned receiver in the class. The second follows because $\Phi$ is increasing: with $t=\ln(1+x)>0$, $t\cosh t-\sinh t>0$.

For the lower bound, the conditioned SLD is at least as informative as the conditioned planet projection. If $b_P=\beta+C+Au$ is the preloss planet-port optical rate, the conditioned and fixed planet-port information obey
\begin{equation}
 \frac{\tau^2}{\tau b_P+\ddet}
 \ge\frac{\tau^2}{b_P+\ddet}.
 \label{eq:lower_pointwise}
\end{equation}
Averaging and using the exact fixed optimum proves the lower bound.
\end{proof}

\subsection{An exclusion independent of direction statistics}\label{sec:direction_free}
The pointwise operator comparison also limits arbitrary direction ensembles. A non-Haar fixed optimum need not be known: it is enough that the planet projection remains an available fixed measurement.

\begin{proposition}[Direction-distribution-independent ceiling]\label{prop:direction_free}
Retain the two-mode rate operator of Eq.~\eqref{eq:reduced}, $C>0$, $A,\beta,\ddet\ge0$, scalar relative transmission $0<\tau\le1$, and direction-independent positive duties. Let the fixed comparator optimize over a class containing $\{P,P_\perp\}$, optionally with recalibration. For any probability law of the recorded leakage direction,
\begin{align}
 \Gwall&\le\drel\tau\Psi(\xi)\le\drel\tau\Psi(x),
 \qquad \xi=\frac{A}{D}\le x,\label{eq:direction_ceiling}\\
 \Psi(\xi)&=\frac{(1+\xi/2)^2}{1+\xi}
           =1+\frac{\xi^2}{4(1+\xi)}.\label{eq:psi}
\end{align}
This upper bound applies to every direction-conditioned finite-outcome single-photon POVM in the declared detector model. It does not require independent direction draws, Haar isotropy, or a particular calibration cadence.
\end{proposition}
\begin{proof}
The unattenuated planet projection has $J_P(s)=1/(D+Au)$. Equations~\eqref{eq:operator_comparison} and~\eqref{eq:zero_floor_pointwise} bound any conditioned receiver by
\begin{align}
 \frac{J_{\rm c}(s)}{J_P(s)}
 &\le\tau\frac{[D+A(1-u)](D+Au)}{D(A+D)}\nonumber\\
 &=\tau\left[1+\frac{\xi^2u(1-u)}{1+\xi}\right]
 \le\tau\Psi(\xi).\label{eq:direction_pointwise}
\end{align}
The last inequality uses $u(1-u)\le1/4$. Multiply by the common nonnegative direction or time weights and sum. The fixed optimum satisfies $J_{\rm f}^{\star}\ge\E_s J_P(s)$, so
\begin{equation}
 \frac{\drel\E_sJ_{\rm c}(s)}{J_{\rm f}^{\star}}
 \le\frac{\drel\E_sJ_{\rm c}(s)}{\E_sJ_P(s)}
 \le\drel\tau\Psi(\xi).
 \label{eq:weighted_direction_bound}
\end{equation}
The first inequality remains valid within each recalibration interval. Finally, $\Psi'(\xi)=\xi(\xi+2)/[4(1+\xi)^2]\ge0$ and $D\ge C$ prove the contrast-only bound.
\end{proof}

Temporal correlations here refer to the \emph{recorded, parameter-ancillary directions}: conditional Poisson information still adds along a correlated history. The argument does not convert unresolved persistent rates into a Poisson likelihood. It also uses the same direction weights for both branches after their direction-independent duty factors are separated. State-selective admission requires the actual weighted comparison; replacing those weights by mean duties need not preserve Eq.~\eqref{eq:direction_ceiling}. For known rates varying between intervals, the same pointwise argument applies using an upper bound on the intervalwise factor $\tau\Psi[A/(C+\beta+\ddet)]$.

The bound has a direct optical interpretation. Its excess over the scalar transmission penalty is limited by $\xi^2u(1-u)$: both leakage components must be present to provide an off-diagonal coherence that a basis change can exploit. The largest relaxed advantage occurs at equal modal powers, $u=1/2$, not when the leakage is parallel or orthogonal to the planet. At weak predictable leakage, $\Psi(\xi)=1+\xi^2/4+O(\xi^3)$, so the available measurement advantage is only second order while attenuation and duty losses remain finite.

At equal contrasts, the ceiling is sharp in the zero-floor limit. Set $A=C$, $\beta=\ddet=0$, and $\ket s=(\ket p+e^{i\varphi}\ket{p_\perp})/\sqrt2$ with uniform relative phase $\varphi$. For a rank-one effect $w\ket v\bra v$ with $a=|\langle p|v\rangle|^2$, direct phase integration gives
\begin{equation}
 \E_\varphi J_m=\frac{w}{C}\frac{a^2}{\sqrt{2a^2+1/4}}
 \le\frac{2wa}{3C}.
 \label{eq:equatorial_fixed}
\end{equation}
Rank-one refinement is valid when $\ddet=0$; completeness gives $\sum_m w_ma_m=1$. Thus $J_{\rm f}^{\star}=2/(3C)$, while $J_{\rm c}=3\tau/(4C)$, giving $\Gwall=9\drel\tau/8$. This is a phase-diverse ensemble without phase-concentration information in the calibration record. A single known static direction instead permits a fixed SLD and does not saturate this bound.

For $\tau=0.80$ and $\drel=0.90$, equal predictable-leakage and planet contrasts imply, for any direction law,
\begin{equation}
 \Gwall\le\frac98\tau\drel=0.8100<1.
 \label{eq:design_distribution_free}
\end{equation}
The Haar optimum is more tightly bracketed:
\begin{equation}
 0.5760\le\Gwall\le\frac{3\tau\drel}{4\ln2}=0.7791.
 \label{eq:design_exclusion}
\end{equation}
Even the optimal conditioned receiver needs at least $23.5\%$ more scheduled time for any direction law, and at least $28.4\%$ more for the Haar ensemble, at the same local information target. A gain at equal contrasts requires $\tau\drel>8/9$; equality merely allows break-even in the ideal limiting ensemble. No change in the nonnegative isotropic optical floor or matched detector rate restores a gain at the quoted losses. The lower bound in Eq.~\eqref{eq:design_exclusion} concerns the Haar optimum, not every implemented measurement.

The direction-free ceiling gives an explicit necessary contrast boundary. For target $G_{\rm tar}$, define the dimensionless $z_{\rm tar}=G_{\rm tar}/(\tau\drel)>1$. Then
\begin{equation}
 x_{\rm nec}=2(z_{\rm tar}-1)+2\sqrt{z_{\rm tar}(z_{\rm tar}-1)}.
 \label{eq:direction_root}
\end{equation}
A gain strictly exceeding $G_{\rm tar}$ requires $x>x_{\rm nec}$; equality in the target permits $x=x_{\rm nec}$. For the quoted losses and targets $G_{\rm tar}=(1,1.25)$, respectively, the necessary boundaries are
\begin{align}
 x_{\rm nec}&\simeq(2.24764,\ 3.73317),\label{eq:direction_roots}\\
 x_{{\rm nec},H}&\simeq(3.2612,\ 5.8059),\label{eq:contrast_roots}
\end{align}
where the second line uses the stronger Haar ceiling of Proposition~\ref{prop:ceiling}. These are exclusions, not sufficient observing conditions. Finite backgrounds tighten them through $\xi<x$. The quantities in Eq.~\eqref{eq:contrast_ray} must be established from modal rates before substituting a spatially averaged instrument contrast.

Two limits clarify the remaining Haar normalization dependence. At fixed $\kappa$ and $x$, $r\to0$ gives
\begin{equation}
 \Gwall\longrightarrow\drel\tau^2\frac{1+\kappa}{\tau+\kappa}.
 \label{eq:fixed_kappa_limit}
\end{equation}
For $\kappa=0.1$, this equals $0.7040$. At the opposite end, taking both $\beta\to0$ and $\ddet\to0$ approaches $\drel\tau\Phi(x)$. Holding a physical $\ddet>0$ fixed instead gives $\kappa\to\infty$, and that ceiling need not be attained. The bounds, rather than a choice of limiting path, establish the exclusion.

\subsection{An explicit, reproducible modal photon budget}\label{sec:photometry}
For a concrete scale, consider a solar analog at distance $d_\star=10\,\mathrm{pc}$, a clear aperture of diameter $D_{\rm tel}=6\,\mathrm m$, and a top-hat channel centered at $\lambda_0=0.95\,\mu\mathrm m$ with dimensionless resolving power $\mathcal R=\lambda_0/\Delta\lambda=140$. Here $\mathrm{pc}$ is the parsec ($1\,\mathrm{pc}\simeq3.08567758\times10^{16}\,\mathrm m$), $\mathrm m$ is the metre, and $1\,\mu\mathrm m=10^{-6}\,\mathrm m$. Adopt stellar temperature $T_\star=5778\,\mathrm K$, where $\mathrm K$ denotes kelvin, radius $R_\star=6.957\times10^8\,\mathrm m$, and dimensionless total throughput $\eta_{\rm opt}=0.10$ including any selection of the analyzed polarization. With $A_{\rm tel}=\pi D_{\rm tel}^2/4$, the detected stellar channel rate is
\begin{align}
 F_\star={}&\eta_{\rm opt}A_{\rm tel}\left(\frac{R_\star}{d_\star}\right)^2
 \int_{\lambda_-}^{\lambda_+}\frac{\pi B_\lambda(T_\star)\lambda}{hc_0}\,d\lambda,
 \label{eq:star_rate}\\
 \lambda_\pm={}&\lambda_0\pm\frac{\lambda_0}{2\mathcal R},\qquad
 B_\lambda(T)=\frac{2hc_0^2}{\lambda^5[\exp(hc_0/\lambda k_BT)-1]}.
 \nonumber
\end{align}
Here $c_0$ is the speed of light in $\mathrm{m\,s}^{-1}$, $h$ is Planck's constant in joule seconds ($\mathrm{J\,s}$), and $k_B$ is Boltzmann's constant in $\mathrm{J\,K}^{-1}$. The symbol $\mathrm J$ here denotes the joule, not Fisher information. With wavelengths in metres, the Planck spectral radiance $B_\lambda$ has units $\mathrm{W\,m}^{-3}\,\mathrm{sr}^{-1}$; watt ($\mathrm W$) means joule per second and steradian ($\mathrm{sr}$) is the solid-angle unit. Dividing radiant power by photon energy $hc_0/\lambda$ gives a photon rate. Equation~\eqref{eq:star_rate} gives $F_\star=1.76878\times10^7\,\mathrm{s}^{-1}$ and, for $\Cplanet=10^{-10}$, $C=1.76878\times10^{-3}\,\mathrm{s}^{-1}$. These choices define an illustration, not an adopted HWO throughput model.

A diffuse background requires a \emph{mode-coupling} prescription. For a uniformly illuminated clear pupil and a normalized accepted pupil mode, Parseval's identity gives the integrated angular response $\lambda^2/A_{\rm tel}$. Thus an ideal locally uniform sky contributes through the equivalent solid angle
\begin{equation}
 \Omega_{\rm mode}=\frac{\lambda_0^2}{A_{\rm tel}}
 =1.35802\times10^{-3}\,\mathrm{arcsec}^2.
 \label{eq:etendue}
\end{equation}
The ratio $\lambda_0^2/A_{\rm tel}$ is a solid angle in steradians; the displayed value converts to square arcseconds, using one arcsecond ($\mathrm{arcsec}$) equal to $\pi/648000$ radians. This is the integral of a single mode's angular acceptance, not a freely chosen circular photometric aperture. Equal normalized modes have equal diffuse coupling under these clear-pupil, uniform-sky assumptions. A coronagraph or a nonuniform sky can violate that equality.

Adopt apparent magnitude $m_{V,\star}=4.83$ in the Johnson $V$ visual band and a gray-scattering, stellar-color extrapolation from that band to the channel. Magnitudes (mag) are dimensionless logarithmic flux measures; $\mu_V$ is the numerical surface brightness in $\mathrm{mag\,arcsec}^{-2}$. The corresponding modal rate is
\begin{equation}
 \frac{\beta_{\rm diffuse}}{F_\star}
 =\left(\frac{\Omega_{\rm mode}}{1\,\mathrm{arcsec}^2}\right)
 10^{-0.4(\mu_V-m_{V,\star})}.
 \label{eq:surface_brightness}
\end{equation}
Exozodiacal light is dust-scattered light in the target system; one exozodi denotes the adopted solar-zodiacal reference brightness, and $N_{\rm ez}$ is its dimensionless multiplier. The Earth-equivalent insolation distance is the orbital radius receiving Earth's stellar irradiance. The values $\mu_{V,\rm ez}=22$ for one exozodi at that radius and $\mu_{V,\rm local}=23$ for dust-scattered zodiacal light in the Solar System follow the explicit exposure-time-calculator conventions of Ref.~\cite{Stark2025ETC}. Their extrapolation here assumes solar colors for this solar analog; they are not claimed to be universal surface brightnesses for every viewing geometry. The blackbody and adopted $V$ magnitude are an approximate solar normalization, not a precision stellar atmosphere or a synthetic Johnson-band calibration.

Equations~\eqref{eq:etendue} and~\eqref{eq:surface_brightness} yield
\begin{align}
 \beta_{\rm ez,1}&=3.25521\times10^{-3}\,\mathrm{s}^{-1}\,\mathrm{mode}^{-1},\nonumber\\
 \beta_{\rm local}&=1.29592\times10^{-3}\,\mathrm{s}^{-1}\,\mathrm{mode}^{-1},\nonumber\\
 \beta(N_{\rm ez})&=\beta_{\rm local}+N_{\rm ez}\beta_{\rm ez,1}
                         +\beta_{\rm other}.
 \label{eq:modal_partition}
\end{align}
We assign $\beta_{\rm other}=10^{-4}\,\mathrm{s}^{-1}$ per mode as an optimistic allowance for other irreducible \emph{optical} terms. This does not demonstrate that finite-star leakage and unpredicted residuals fit that allowance. The detector rate is separately fixed at $\ddet=2\times10^{-4}\,\mathrm{s}^{-1}$ per output in every background comparison. All additional leakage $A$ is treated as predictable, an optimistic assumption relaxed in Sec.~\ref{sec:temporal}.

\begin{table*}[t]
\caption{Illustrative physical-rate budget from Eq.~\eqref{eq:modal_partition}, with $\Cplanet=10^{-10}$, $\tau=0.80$, $d_{\rm c}/d_{\rm f}=0.90$, and the same detector rate $\ddet=2\times10^{-4}\,\mathrm{s}^{-1}$ per output in all rows. Thresholds solve Eq.~\eqref{eq:exact_gain} using the stated modal coupling; they are not mission exposure predictions.}
\label{tab:budget}
\centering\small
\begin{tabular}{ccccccc}
\toprule
$N_{\rm ez}$&$\beta$ ($10^{-3}\,\mathrm{s}^{-1}$ per mode)&$\gamma$&$\kappa$&$G$ at $\Cpred=10^{-10}$&$\Cpred$ at $G=1$&$\Cpred$ at $G=1.25$\\
\midrule
1&4.6511&0.3803&0.04300&0.7190&$1.9102\times10^{-9}$&$3.4147\times10^{-9}$\\
3&11.1616&0.1585&0.01792&0.7184&$3.9278\times10^{-9}$&$7.0456\times10^{-9}$\\
10&33.9480&0.0521&0.00589&0.7192&$1.0977\times10^{-8}$&$1.9733\times10^{-8}$\\
\bottomrule
\end{tabular}
\end{table*}

Table~\ref{tab:budget} and Fig.~\ref{fig:contrast} distinguish both proved ceilings from the Haar finite-budget thresholds. One exozodi gives $\gamma=0.3803$; $\gamma=0.5$ would correspond to $0.658$ exozodi under this partition. At equal contrasts, the three gains are close to $0.719$, or about 39\% longer wall time. The one-exozodi 20\%-saving threshold is the dimensionless contrast $\Cpred=3.41\times10^{-9}$, about 34 times the assigned planet contrast, and larger diffuse backgrounds move the threshold further into leakage-dominated operation. Increasing background does not restore an advantage at the equal-contrast point; the small nonmonotonic variation of those three below-unity ratios reflects the competing optical and fixed-detector contributions, not a new beneficial regime.

\begin{figure*}[t]
\centering
\begin{minipage}{0.485\textwidth}\centering\includegraphics[width=\linewidth]{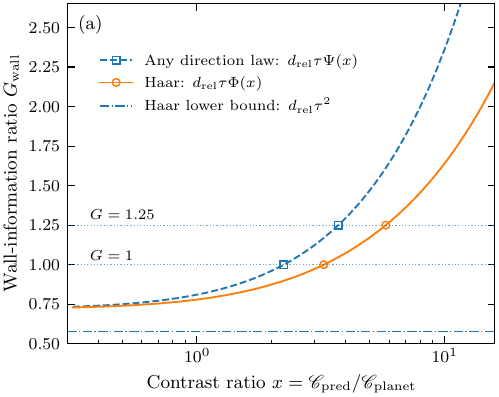}\end{minipage}\hfill
\begin{minipage}{0.485\textwidth}\centering\includegraphics[width=\linewidth]{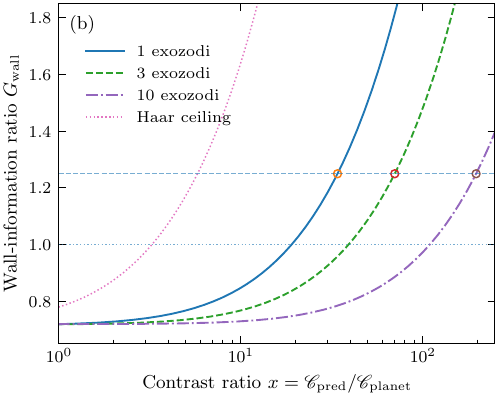}\end{minipage}
\caption{Dimensionless contrast limits and finite-background calculations for $\tau=0.80$, $\drel=0.90$. (a) The direction-distribution-independent ceiling $\drel\tau\Psi(x)$ and the stronger Haar ceiling $\drel\tau\Phi(x)$ both exclude a gain at $x=1$, with upper values $0.8100$ and $0.7791$. The lower bound $\drel\tau^2$ applies to the Haar optimum. Open squares and circles mark the respective break-even and 20\%-saving boundaries: $(2.24764,3.73317)$ without a direction law and $(3.2612,5.8059)$ for Haar directions. These boundaries are necessary, not sufficient. (b) Exact Haar gains for the three modal budgets in Table~\ref{tab:budget}, with the physical detector rate held fixed. Open markers show 20\%-saving crossings; the dotted curve is the Haar ceiling. Finite optical backgrounds impose substantially stronger restrictions.}
\label{fig:contrast}
\end{figure*}

A realistic coronagraph changes extended-source throughput, planet acceptance, and leakage coupling differently. Disk inclination and structure can also create anisotropy and astrophysical nuisance parameters. The budget is therefore a transparent normalization against which to test a future field calculation, while the direction-distribution-independent exclusion remains valid for an isotropic floor without adopting this particular normalization.

\section{Prediction accuracy and directional diversity}\label{sec:temporal}
\subsection{The residual must be predictable, not merely bright}
Let $L$ be the total post-control leakage in the selected pair before separating its predictable part. A simple isotropic prediction-error model gives
\begin{equation}
 A=(1-\epsilon)L,\quad \beta=b+\frac{\epsilon L}{2},\quad
 r=\frac{(1-\epsilon)q}{1+\epsilon q/2},\quad q=L/b,
 \label{eq:mismatch}
\end{equation}
where $b$ is the other irreducible optical rate per mode and $0\le\epsilon\le1$ is an error \emph{power fraction}, not a root-mean-square amplitude error. Both $\gamma=C/\beta$ and $\kappa=\ddet/\beta$ change with the same denominator when $C$ and $\ddet$ are fixed. The decomposition assumes that error power acts as an incoherent isotropic floor on the counting timescale. Persistent conditional-rate uncertainty is not made Poisson merely by averaging it; it is treated separately in Sec.~\ref{sec:latent}.

At fixed $\epsilon>0$, increasing raw leakage cannot make $r$ arbitrarily large:
\begin{equation}
 r\le\frac{2(1-\epsilon)}{\epsilon}.
 \label{eq:mismatch_cap}
\end{equation}
For example, the zero-continuum, zero-detector, unit-relative-duty case with $\tau=0.8$ needs $r>3.7479$ for break-even and $r\ge7.6377$ for a 20\% saving. Equation~\eqref{eq:mismatch_cap} excludes these targets for $\epsilon\ge0.3480$ and $0.2075$, respectively. These are sensitivity slices, not permissible errors for every HWO block. With $q=25$, $C/b=1.125$, $\ddet/b=0.225$, $\tau=0.8$, and $\drel=0.9$, even the Haar limit requires $\epsilon<0.14325$ for a strict gain. At $\epsilon=0.10$ this assignment gives $(r,\gamma,\kappa)=(10,0.5,0.1)$.

\subsection{A recalibrated fixed comparator}
If the residual direction is static and known, a fixed SLD reaches the same preloss quantum information as a conditioned SLD. A useful receiver therefore needs directional variation as well as prediction. To quantify this distinction, write $S=(I+\bm n\cdot\bm\sigma)/2$, where $\bm\sigma$ is the vector of Pauli matrices and the Bloch vector $\bm n$ is a dimensionless unit vector encoding relative modal amplitudes and phase. Let the direction within a fixed calibration interval have the conditional density
\begin{align}
 p_\zeta(\bm n\mid\bm n_0)&=\frac{\zeta e^{\zeta\bm n\cdot\bm n_0}}{4\pi\sinh\zeta},\nonumber\\
 \mathcal L(\zeta)\equiv\coth\zeta-\frac1\zeta
 &=e^{-T_{\rm cal}/t_{\rm coh}}.
 \label{eq:vmf}
\end{align}
The nonnegative concentration parameter $\zeta$ and its mean resultant $\mathcal L(\zeta)$ are dimensionless; $p_\zeta$ is a density per unit solid angle on the Bloch sphere. The directional coherence parameter $t_{\rm coh}$ and calibration interval $T_{\rm cal}$ are both in seconds. The mean directions $\bm n_0$ are Haar distributed across calibration intervals, so the unconditional direction remains Haar. The second relation is a declared moment-matching model: it equates the mean resultant length of the von Mises--Fisher density to an exponential correlation. It does not derive a disturbance process, establish a causal predictor, or make a one-point density into a physical time series.

The fixed receiver is allowed exact knowledge of $\bm n_0$ from its calibration, may choose any fixed POVM for that interval, and retains each realized $s$ in its likelihood. This deliberately strong comparator isolates the effect of slow directional concentration. Its calibration time enters $d_{\rm f}$, while its calibration error is optimistically zero. The conditioned branch knows the predicted instantaneous $s$ in the reduced model. Let $J_{{\rm f},H}^{\star}$ and $J_{{\rm c},H}$ denote the corresponding admitted Haar information, and set
\begin{equation}
 G_H=\drel J_{{\rm c},H}/J_{{\rm f},H}^{\star}.
 \label{eq:haar_gain}
\end{equation}
This number alone is not the ratio against the recalibrated fixed receiver.

\subsection{A sufficient bound against arbitrary fixed POVMs}
The concentration model admits a useful analytic lower bound on the \emph{actual} gain. Relative to Haar measure $dH=d\Omega/(4\pi)$, its density obeys
\begin{equation}
 w_\zeta(\bm n\mid\bm n_0)\le M_\zeta,
 \qquad M_\zeta=\frac{2\zeta}{1-e^{-2\zeta}},\qquad M_0=1.
 \label{eq:density_max}
\end{equation}
For every fixed POVM and every mean direction,
\begin{align}
 \int w_\zeta J(s;\Pi)\,dH
 &\le M_\zeta\int J(s;\Pi)\,dH\nonumber\\
 &\le M_\zeta J_{{\rm f},H}^{\star}.
 \label{eq:density_proof}
\end{align}
Optimization over arbitrary finite-outcome fixed measurements and averaging over $\bm n_0$ preserve this inequality. The conditioned numerator remains $J_{{\rm c},H}$ because its marginal ensemble is Haar. If $J_{{\rm f},H}^{Q}$ is the Haar average of the pointwise fixed-branch quantum bound, then the true recalibrated fixed optimum satisfies both upper bounds. Hence
\begin{equation}
 G^{\star}_{\rm wall}\ge
 \max\!\left\{\frac{G_H}{M_\zeta},\,
 \drel\frac{J_{{\rm c},H}}{J_{{\rm f},H}^{Q}}\right\}.
 \label{eq:diversity_lower}
\end{equation}
Unlike a comparison only with selected fixed projective bases, Eq.~\eqref{eq:diversity_lower} is a sufficient gain certificate within the stipulated ensemble and rate model.

For the $\epsilon=0.10$ assignment above, $G_H=1.071528$ and the quantum-bound term is $0.7106$. Solving $M_\zeta=G_H/G_{\rm tar}$ gives sufficient directional conditions
\begin{align}
 t_{\rm coh}/T_{\rm cal}&<0.26598 &&(G^{\star}_{\rm wall}>1),\nonumber\\
 t_{\rm coh}/T_{\rm cal}&\le0.19514 &&(G^{\star}_{\rm wall}\ge1/0.95).
 \label{eq:sufficient_temporal}
\end{align}
The second certifies at least a 5\% local exposure saving. No choice of diversity can make this stress point meet the 20\% target, since its Haar gain is only $1.0715$. At the more favorable final row of Table~\ref{tab:penalties}, however, the same density argument certifies $G\ge1.25$ whenever $t_{\rm coh}/T_{\rm cal}\le0.32729$. This is a joint rate-and-diversity example in the concentration model, not evidence that an HWO predictor realizes those parameters.

\subsection{A constructive upper bound and the remaining interval}
A feasible fixed measurement gives a complementary upper bound on gain. We search over a projective axis $\bm h$ described by polar angle $\theta$ and azimuth $\phi$ relative to the planet Bloch axis. All angular coordinates are in radians. Let $\bm p$ be the dimensionless unit Bloch vector of $P$, and put $\mu_0=\bm n_0\cdot\bm p\in[-1,1]$. Then
\begin{equation}
 \delta_0=\mu_0\cos\theta+\sqrt{1-\mu_0^2}\sin\theta\cos\phi.
 \label{eq:axis_geometry}
\end{equation}
For $z=\bm n\cdot\bm h$, integration over the azimuth about $\bm h$ gives
\begin{equation}
 f_\zeta(z\mid\delta_0)=\frac{\zeta e^{\zeta\delta_0z}}{2\sinh\zeta}
 I_0\!\left(\zeta\sqrt{1-\delta_0^2}\sqrt{1-z^2}\right),
 \label{eq:vmf_marginal}
\end{equation}
Here $z$ is a dimensionless Bloch cosine, distinct from the photon-mode overlap $u$; $I_0$ is the modified Bessel function of the first kind and order zero, and $f_0=1/2$. With scalar overlap $a_\theta=(1+\cos\theta)/2$ and $b_0=\beta+\ddet$, the fixed-axis information is
\begin{align}
 J_{\rm f}(\theta,\phi\mid \mu_0)=\int_{-1}^1 &f_\zeta(z\mid\delta_0)
 \left[\frac{a_\theta^2}{b_0+Ca_\theta+A(1+z)/2}\right.\nonumber\\[-2pt]
 &\left.+\frac{(1-a_\theta)^2}{b_0+C(1-a_\theta)+A(1-z)/2}\right]dz.
 \label{eq:axis_objective}
\end{align}
We retain the best axis found by a two-angle scan and local refinement, then average over uniform $\mu_0\in[-1,1]$. Call the resulting constructive information $J_{\rm f}^{\rm con}$. It need not be the global projective maximum to give the valid ordering
\begin{equation}
 G^{\star}_{\rm wall}\le
 \drel\frac{J_{{\rm c},H}}{J_{\rm f}^{\rm con}}.
 \label{eq:constructive_upper}
\end{equation}
A better feasible fixed measurement only tightens this upper bound. We do not label a local optimizer as a global POVM proof.

\begin{figure}[t]
\centering\includegraphics[width=\columnwidth]{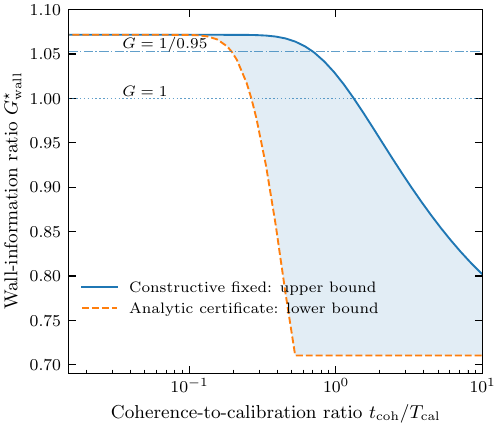}
\caption{Directional-diversity bounds for the concentration model at $(r,\gamma,\tau,\kappa,\drel)=(10,0.5,0.8,0.1,0.9)$. The lower curve is the analytic arbitrary-fixed-POVM certificate in Eq.~\eqref{eq:diversity_lower}; the upper curve uses the feasible two-angle fixed measurement in Eq.~\eqref{eq:constructive_upper}. The shaded interval contains the true gain. The 5\%-saving lower and upper crossings are $0.19514$ and $0.69974$; the break-even crossings are $0.26598$ and $1.32950$. The bounds meet in the Haar limit and approach the transmission/detector penalty as a fixed basis recovers the direction. The mapping to coherence time is model-specific.}
\label{fig:temporal}
\end{figure}

The numerical objective is fully specified by Eqs.~\eqref{eq:axis_geometry}--\eqref{eq:axis_objective}. Axis inversion and reflection permit $0\le\theta\le\pi/2$ and $0\le\phi\le\pi$. We use a $21\times11$ grid, retain its best three starts, and refine with bounded two-dimensional optimization. The mean-direction and conditional integrals use 32 and 128 Gauss--Legendre nodes; 48 and 192 nodes provide a convergence check. The grid candidate is always retained, so refinement cannot lower the constructed information. Exponentially scaled Bessel evaluation avoids overflow.

Figure~\ref{fig:temporal} combines the analytic lower and constructive upper bounds. For the same stress point, the upper curve crosses break-even at $t_{\rm coh}/T_{\rm cal}=1.32950$ and 5\% saving at $0.69974$. The model therefore certifies a 5\% saving below $0.19514$, leaves the intermediate interval unresolved, and excludes it above $0.69974$ along the evaluated slice. The crossings are quadrature-converged numerical results for a constructive comparator, not universal instrumental constants. 

\subsection{Timing and information must hold together}
Directional diversity does not establish prediction accuracy. In Eq.~\eqref{eq:vmf}, $t_{\rm coh}$ parameterizes a conditional direction distribution; it is not automatically a forecast horizon. The operational requirement is Eq.~\eqref{eq:causality} together with an acceptable held-out error throughout the admitted interval. Under the isotropic error model, this can be stated as
\begin{equation}
 \epsilon_{\max}=\sup_{h\in[t_{\rm ff},\,t_{\rm ff}+T_{\rm hold}]}\epsilon(h)
 \le\epsilon_{\rm allowed},
 \label{eq:horizon_error}
\end{equation}
where the allowed error must be evaluated jointly with the rates and duties, rather than assigned independently. Predictable periodic dynamics can remain forecastable beyond an ordinary correlation time.

One may additionally adopt a conservative coherence-limited screen $\max(T_{\rm upd},t_{\rm ff})<t_{\rm coh}$. Within that operational convention, compatibility with the constructive 5\%-saving exclusion requires
\begin{equation}
 \frac{\max(T_{\rm upd},t_{\rm ff})}{T_{\rm cal}}<0.69974.
 \label{eq:necessary_window}
\end{equation}
This is a model- and screen-dependent compatibility condition, not a universal necessary condition for predictive control or a sufficient hardware criterion.

For an assigned example, take $T_{\rm cal}=60\,\mathrm{s}$, $T_{\rm upd}=T_{\rm hold}=0.10\,\mathrm{s}$, and $t_{\rm ff}=0.02\,\mathrm{s}$. The prediction is required through a $0.12\,\mathrm{s}$ horizon. At the stipulated error fraction $\epsilon=0.10$, the sufficient diversity condition for 5\% saving is $t_{\rm coh}\le11.71\,\mathrm{s}$, whereas the constructive exclusion begins at approximately $41.98\,\mathrm{s}$.

A moment-level forecast check clarifies the required predictor. If the exponential correlation is additionally interpreted as an isotropic Markov direction model, persistence of an exactly known last direction gives mean Bloch alignment $R_h=e^{-h/t_{\rm coh}}$. At $t_{\rm coh}=1\,\mathrm{s}$ and $h=0.12\,\mathrm{s}$, $R_h=0.8869$, or mean direction infidelity $(1-R_h)/2=0.05654$. The corresponding projector mean is $R_hS_{\rm last}+(1-R_h)I/2$, so Eq.~\eqref{eq:mismatch} uses $\epsilon=1-R_h=0.1131$, not the infidelity itself. Under this additional Markov assumption, $\epsilon\le0.10$ requires a horizon at most $0.10536\,\mathrm{s}$ or $t_{\rm coh}\ge1.139\,\mathrm{s}$ at the same horizon. Retaining both $1\,\mathrm{s}$ and $0.12\,\mathrm{s}$ instead requires additional predictable structure beyond this persistence model; Eq.~\eqref{eq:vmf} alone supplies no such predictor.

Faster updates can instead become limited by switching dead time: without overlap, a blanking time $t_{\rm blank}$ per update alone gives $d_{\rm c}\le1-t_{\rm blank}/T_{\rm upd}$. Meeting a fast timing inequality while losing the required duty is not an eligible operating point.

\section{Spectroscopic inference and receiver selection}\label{sec:spectroscopy}
\subsection{Spatial information does not by itself identify a molecule}
A molecular amplitude cannot be separated from an unknown continuum using one spatial mode in one spectral block when both perturb the same rate. Cross-wavelength information is essential, as in established quantum-limit and template-based spectroscopy analyses~\cite{Huang2023,Santamaria2025,Ruffio2026}. Let $a$ index simultaneous wavelength--polarization blocks, $\vartheta$ be a dimensionless signed feature amplitude, and $\nu$ a dimensionless shared-continuum perturbation. For receiver $j$, define
\begin{align}
 R_{j,a}(\vartheta,\nu\mid s)
 ={}&\tau_{j,a}\beta_a
 [I+r_aS+\gamma_aP\nonumber\\
 &+(m_a\vartheta+q_a\nu)P],
 \label{eq:spectral_rate}
\end{align}
with $\tau_{{\rm f},a}=1$ and $\tau_{{\rm c},a}=\tau_a$. Detector rates are added per output after this optical rate operator. The dimensionless templates $m_a$ and $q_a$ are fixed in advance; rates are positive under the stated parameter domain. All bases are designed at the fiducial $\vartheta=\nu=0$, rather than using the unknown realized parameter to select a measurement.

For implemented effects, the admitted spatial kernel is
\begin{equation}
 K_{j,a}=\E_s\sum_o
 \frac{[\tau_{j,a}\beta_a\Tr(P\Pi_{j,a,o})]^2}
 {\Tr[R_{j,a}(0,0\mid s)\Pi_{j,a,o}]+d_{{\rm det},a}}.
 \label{eq:spectral_kernel}
\end{equation}
The index $o$ labels physical detector outputs. Because $\vartheta$ and $\nu$ are dimensionless, $K_{j,a}$ has units $\mathrm{s}^{-1}$. In the reduced model it equals $\beta_a^2J_{j,a}$ from Eq.~\eqref{eq:detector_rates}. With block-independent counts conditional on the shared parameters, the information rate and its continuum-profiled value are
\begin{align}
 F_j&=\sum_a K_{j,a}\bm v_a\bm v_a^T,
 \quad\bm v_a=(m_a,q_a)^T,\label{eq:fisher_spectrum}\\
 J_{\vartheta\mid\nu}^{(j)}
 &=F_{\vartheta\vartheta}^{(j)}
 -\frac{[F_{\vartheta\nu}^{(j)}]^2}{F_{\nu\nu}^{(j)}}.
 \label{eq:schur}
\end{align}
The matrix $F_j$ and the continuum-profiled information rate $J_{\vartheta\mid\nu}^{(j)}$ also have units $\mathrm{s}^{-1}$; $tF_j$ is dimensionless. The superscript $T$ denotes transpose. Absolute duties multiply the corresponding $K_{j,a}$ when wall time is the resource. We initially take no external continuum information. A fixed external calibration instead contributes a total matrix $F_{\rm cal}$ to $tF_j$; it must not be treated as a science-time-proportional information rate unless the calibration effort is scaled with exposure.

The determinant and profiled information can be written
\begin{align}
 \det F_j&=\sum_{a<b}K_{j,a}K_{j,b}(m_aq_b-m_bq_a)^2,
 \label{eq:det_spectral}\\
 J_{\vartheta\mid\nu}^{(j)}
 &=\frac{\sum_{a<b}K_{j,a}K_{j,b}(m_aq_b-m_bq_a)^2}
         {\sum_aK_{j,a}q_a^2}.
 \label{eq:schur_pairwise}
\end{align}
These follow by expanding the $2\times2$ determinant. The feature is identifiable only if the molecular and continuum templates are not proportional on the informative blocks. A large spatial information ratio cannot cure a spectral degeneracy: when the two templates approach proportionality, both receivers' absolute profiled information can tend to zero even while their ratio remains finite. Additional nuisance templates extend this rank and conditioning requirement; stochastic correlations between blocks require more than the block-additive likelihood used here.

\subsection{Template bounds and the role of independent gates}
Let $\rho_a=K_{{\rm c},a}/K_{{\rm f},a}$, with extrema $\rho_-$ and $\rho_+$. Since every $\bm v_a\bm v_a^T$ is positive semidefinite,
\begin{equation}
 \rho_-F_{\rm f}\preceq F_{\rm c}\preceq\rho_+F_{\rm f}.
 \label{eq:loewner}
\end{equation}
The representation
\begin{equation}
 J_{\vartheta\mid\nu}(F)=\min_z(1,z)F(1,z)^T
 \label{eq:schur_variational}
\end{equation}
proves monotonicity under this order and homogeneity of degree one. Therefore
\begin{equation}
 \rho_-\le\frac{J_{\vartheta\mid\nu}^{({\rm c})}}
                      {J_{\vartheta\mid\nu}^{({\rm f})}}\le\rho_+.
 \label{eq:template_bounds}
\end{equation}
The interval is sharp over unrestricted identifiable template pairs: one may put a molecular score in an extremal-ratio block and an independent continuum score in another. It need not be sharp for a prescribed physical spectrum. A common, fixed external calibration preserves monotonicity at fixed exposure but generally changes the homogeneous ratio interval to $[\min(1,\rho_-),\max(1,\rho_+)]$.

The contrast exclusion also propagates through this nuisance fit. If every participating block has $x_a=1$ and the loss budget of Eq.~\eqref{eq:design_distribution_free}, its duty-charged conditioned kernel is at most $0.81$ of the corresponding planet-projection kernel. Summing these rank-one contributions and using Eq.~\eqref{eq:schur_variational} gives the same upper ratio for every identifiable feature--continuum template pair without external calibration. This statement uses the bound in \emph{every} block; one favorable or excluded spatial block cannot determine the complete spectral gain.

For independent channel gates, define the wall kernels
\begin{equation}
 \widetilde K_a=\max\{d_{{\rm f},a}K_{{\rm f},a},\,
                         d_{{\rm c},a}K_{{\rm c},a}\}.
 \label{eq:fallback}
\end{equation}
Any other blockwise fixed/conditioned route has a no-larger kernel in every block. Consequently its Fisher matrix is no larger in positive-semidefinite order, and Eq.~\eqref{eq:schur_variational} proves that Eq.~\eqref{eq:fallback} maximizes the local profiled information. This corollary of block additivity guarantees no loss relative to all-fixed routing when kernels and costs are known. The same argument works for several nuisance amplitudes by replacing $z$ in Eq.~\eqref{eq:schur_variational} with a vector, and for a common calibration matrix added to both routes.

A common instrument gate changes the optimization. If any conditioned channel imposes one shared duty loss on every channel, independent maximization of \emph{duty-charged} kernels is invalid. For a common gate-on duty independent of the selected subset, first select the best admitted kernels within that regime, then compare its complete Fisher matrix or profiled information with full-duty all-fixed operation. If overhead depends on the subset, update frequency, or previous commands, the actual coupled scheduling problem must be solved instead.

\subsection{A signed absorption example with a shared continuum}
Consider three synthetic blocks centered at $0.925$, $0.945$, and $0.965\,\mu\mathrm m$, with common $\beta_a=0.002\,\mathrm{s}^{-1}$ per mode and $d_{{\rm det},a}=0.0002\,\mathrm{s}^{-1}$ per output. Their other parameters and computed kernels are given in Table~\ref{tab:spectral_blocks}. These are assigned channel values, not a propagated molecular atmosphere or instrument spectrum. Choose a central absorption feature and a flat unknown continuum,
\begin{equation}
 \bm m=(0,-1,0),\qquad \bm q=(1,1,1).
 \label{eq:templates}
\end{equation}
The two sidebands now matter even though their direct molecular derivative is zero: they estimate the continuum against which the absorption is measured. For arbitrary effective kernels, put $S_K=K_1+K_3$. Equation~\eqref{eq:schur} reduces to
\begin{equation}
 J_{\vartheta\mid\nu}=\frac{K_2S_K}{K_2+S_K}.
 \label{eq:absorption_information}
\end{equation}
This is the inverse-variance structure of comparing the central channel with an independently estimated reference continuum. Better central-band precision eventually becomes limited by $S_K$, and better continuum precision eventually becomes limited by $K_2$. It explains why discarding a nominally ``feature-free'' channel can damage the molecular measurement.

\begin{table*}[t]
\caption{Assigned heterogeneous blocks and admitted kernels for the absorption example. Rates are physical; $K$ has units $\mathrm{s}^{-1}$. All rows use $\beta=0.002\,\mathrm{s}^{-1}$ per mode and the same per-output detector rate $0.0002\,\mathrm{s}^{-1}$. The last column charges 90\% conditioned duty and full fixed duty.}
\label{tab:spectral_blocks}
\centering\small
\begin{tabular}{cccccccc}
\toprule
$\lambda$ ($\mu\mathrm m$)&$r_a$&$\gamma_a$&$\tau_a$&$K_{{\rm f},a}$ ($10^{-4}\,\mathrm{s}^{-1}$)&$K_{{\rm c},a}$ ($10^{-4}\,\mathrm{s}^{-1}$)&$\rho_a$&$0.9\rho_a$\\
\midrule
0.925&7&0.25&0.752&5.20616&5.34416&1.02651&0.92386\\
0.945&10&0.50&0.800&3.96200&4.71711&1.19059&1.07153\\
0.965&12&0.70&0.820&3.39480&4.30538&1.26823&1.14140\\
\bottomrule
\end{tabular}
\end{table*}

Figure~\ref{fig:routing} shows the local decision: at 90\% conditioned duty, keep the first block fixed and condition the other two. Table~\ref{tab:spectral_policies} propagates those choices through the shared-continuum fit. Uniform conditioned operation has local wall gain $1.0512$, whereas channel-local fallback gives $1.0665$, or 6.24\% shorter asymptotic exposure. The fixed and fallback profiled informations are $2.71250\times10^{-4}$ and $2.89294\times10^{-4}\,\mathrm{s}^{-1}$. The full template interval is $[0.92386,1.14140]$ for uniform duty-charged conditioning and $[1,1.14140]$ for independent fallback. Thus uniform conditioning can hurt another identifiable template even though it helps the declared one.

\begin{figure}[t]
\centering\includegraphics[width=\columnwidth]{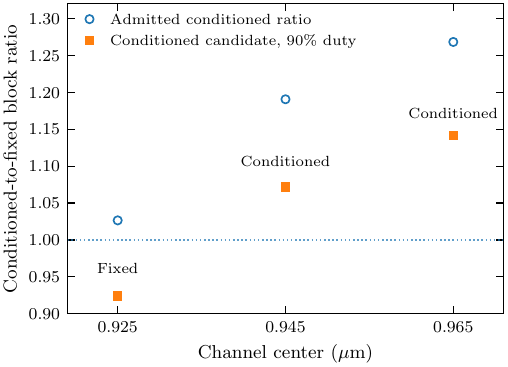}
\caption{Receiver selection for the three discrete blocks in Table~\ref{tab:spectral_blocks}. Admitted ratios exceed unity in all three, but the first channel loses after its own conditioned duty is charged. Independent gating therefore selects fixed, conditioned, conditioned. No interpolation between wavelengths is assumed. The first channel still contributes to absorption precision through its continuum constraint, as made explicit by Eq.~\eqref{eq:absorption_information}.}
\label{fig:routing}
\end{figure}

\begin{table}[t]
\caption{Nuisance-adjusted information for Eq.~\eqref{eq:templates}, with $J$ in $10^{-4}\,\mathrm{s}^{-1}$. The continuum penalty is the fractional subtraction from the raw $F_{\vartheta\vartheta}$ in the same row. ``Admitted'' excludes a duty penalty; the other nonreference rows use the specified wall-time topology.}
\label{tab:spectral_policies}
\centering\small\setlength{\tabcolsep}{4pt}
\begin{tabular}{lrrr}
\toprule
Policy&$J$&Penalty (\%)&$G$\\
\midrule
All fixed&2.71250&31.54&1.0000\\
All conditioned, admitted&3.16831&32.83&1.1680\\
All conditioned, 90\% duty&2.85147&32.83&1.0512\\
Independent fallback&2.89294&31.86&1.0665\\
Mixed route, common gate&2.83796&33.15&1.0463\\
\bottomrule
\end{tabular}
\end{table}

With one common 90\% instrument gate, all three admitted conditioned kernels exceed their fixed counterparts. The optimal gate-on route is therefore all conditioned, not the mixed route. Its gain $1.0512$ exceeds the mixed common-gate value $1.0463$, and both must still be compared with the all-fixed value one. The topology determines the correct reference resource.

A normalized curvature template differs from the absorption template only by a scale factor and a continuum component:
\begin{equation}
 \frac{(1,-2,1)}{\sqrt6}
 =\sqrt{\frac32}(0,-1,0)+\frac{(1,1,1)}{\sqrt6}.
 \label{eq:template_reparameterization}
\end{equation}
The free continuum absorbs the last term. For every set of kernels, the curvature template therefore has $3/2$ times the profiled information of Eq.~\eqref{eq:templates} and the same receiver gain.

\subsection{Finite-count absorption with the continuum profiled}
We test a finite signed perturbation with a jointly fitted continuum in the same three-block model. Set $\vartheta_1=0.12$, true $\nu=0$, and scheduled duration $t=3.0\times10^6\,\mathrm{s}$ ($833.3\,\mathrm h$, with one hour $\mathrm h=3600\,\mathrm s$). Both fixed and independent-fallback routes use the same pre-receiver orientation ensemble and physical rates from Table~\ref{tab:spectral_blocks}; their incompatible measurements use independent photon draws. The SLD bases are frozen at the fiducial continuum and null feature.

The orientation integral is represented by 32 Gauss--Legendre nodes $u_i$ with weights $w_i$ summing to one. Only $u$ is needed because the SLD output rates and planet overlaps are invariant under the leakage phase rotation in this model. These nodes specify deterministic fractional dwell times, not a sampled temporal WFS forecast. Let $b_{aio}$ and $g_{aio}=\tau_{j,a}\beta_a\Tr(P\Pi_{j,a,o})$ be the physical null rate and amplitude slope in each output. The cell means are
\begin{equation}
 \mu_{aio}(\vartheta,\nu)
 =t\,d_{j,a}w_i[b_{aio}+g_{aio}(m_a\vartheta+q_a\nu)].
 \label{eq:finite_cells}
\end{equation}
Both hypotheses fit the same nuisance interval $\nu\in[-0.2,0.2]$. All optical component rates remain nonnegative over this domain for the two tested feature values. The test statistic is the difference of profiled Poisson log likelihoods,
\begin{align}
 \ell(\vartheta,\nu)&=\sum_{aio}[n_{aio}\ln\mu_{aio}(\vartheta,\nu)
                               -\mu_{aio}(\vartheta,\nu)],\nonumber\\
 T&=\max_\nu\ell(\vartheta_1,\nu)-\max_\nu\ell(0,\nu).
 \label{eq:profile_statistic}
\end{align}
The omitted factorial terms cancel. Each one-dimensional nuisance likelihood is concave and is solved by bracketed Newton iterations; a numerical boundary is not substituted for a prior distribution.

For each route, three independent streams of $10^5$ trials respectively calibrate the upper 1\% null threshold, evaluate held-out null events, and evaluate alternative detections. The seed is 20260923, with independent seed sequences indexed by route and stream. Calibration and evaluation use true $\nu=0$; the procedure does not claim uniform false-alarm control over every nuisance value. The false-alarm probability $P_{\rm FA}$ and detection probability $P_{\rm D}$ are dimensionless fractions under the null and alternative, respectively. Table~\ref{tab:finite} reports them with conditional 95\% Wilson intervals. No profile estimate reached a boundary. This is a finite-count calculation of the specified absorption coordinate and shared continuum, not a forecast for a named molecule or atmospheric retrieval.

\begin{table*}[t]
\caption{Finite-count signed-absorption test at the same scheduled duration $3.0\times10^6\,\mathrm{s}$ and $\vartheta_1=0.12$, true $\nu=0$. Each route uses independent $10^5$-trial threshold-calibration, null-evaluation, and alternative-evaluation streams. Parentheses are 95\% Wilson intervals conditional on the calibrated threshold; they exclude threshold-calibration and physical-model uncertainty. The last column is the local one-sided Wald approximation, not the finite Poisson calculation.}
\label{tab:finite}
\centering\small
\begin{tabular}{lccc}
\toprule
Route&Held-out $P_{\rm FA}$&$P_D$&Local Wald $P_D$\\
\midrule
All fixed&$0.01048\ (0.00987,0.01113)$&$0.87114\ (0.86905,0.87320)$&0.86364\\
Independent fallback&$0.01009\ (0.00949,0.01073)$&$0.89479\ (0.89287,0.89668)$&0.88664\\
\bottomrule
\end{tabular}
\end{table*}

The local check is $P_D\simeq\Phi_N[\vartheta_1\sqrt{tJ_{\vartheta\mid\nu}}-\Phi_N^{-1}(0.99)]$, where $\Phi_N$ is the standard-normal distribution function. The exact-profile simulation differs by less than $0.009$ in detection probability for the two routes. Equality is not expected at a finite 24\% reduction of the central planet continuum. In particular, the ratio of the two detection probabilities is not an exposure-time gain. The comparison shows that the signed inference and fallback can be implemented in a calibrated count likelihood at the stated nuisance value; broader nuisance coverage and field uncertainty remain separate requirements.

\section{Physical applicability and limitations}\label{sec:discussion}
\subsection{From controlled fields to the selected modal pair}
The results concern a predictable \emph{post-control} residual. Given a held-out WFS history, an optical propagation should supply a conditional mean leakage field $\widehat{\bm e}$ and its error covariance $\Sigma_e$, in photon-rate-normalized modal coordinates: $\bm e$ and $\widehat{\bm e}$ have units $\mathrm{s}^{-1/2}$, and $\Sigma_e$ has units $\mathrm{s}^{-1}$. A natural decomposition is
\begin{equation}
 \E[\bm e\bm e^\dagger\mid y]
 =\widehat{\bm e}\widehat{\bm e}^{\dagger}+\Sigma_e,
 \qquad A=\|\widehat{\bm e}\|^2,\quad
 \ket s=\widehat{\bm e}/\sqrt A.
 \label{eq:conditional_field}
\end{equation}
Finite stellar diameter, diffuse emission, unpredicted wavefront error, polarization leakage, and pointing uncertainty contribute additional generally nonisotropic rate operators. If the selected pair contains one dominant predictable mode and an approximately isotropic irreducible floor, Eq.~\eqref{eq:reduced} may be adequate. An aligned unequal floor is covered by Theorem~\ref{thm:fixed}. A nonaligned covariance, several important predictable modes, or a mixed planet operator requires direct multimode optimization. Nonunitary ``whitening'' of a background is not a lossless change of coordinates and cannot be used to import the Haar theorem without charging its effect on photons, signal, and direction statistics.

The modal quantities must be estimated jointly. A field can be bright but unpredictable, accurately predicted but almost static, or diverse only in intervals lost to switching. Separate favorable marginal values of $r$, $\epsilon$, $t_{\rm coh}/T_{\rm cal}$, and throughput do not establish a common eligible subset. The independent-gate rule also needs held-out kernels or conservative confidence bounds; selecting a route using the same upward noise fluctuation later used to claim its gain is not covered by Eq.~\eqref{eq:fallback}.

\subsection{What timing and power spectra can establish}
A possible physical opportunity is residual content that is observable and predictable but not sufficiently rejected by the common wavefront controller. It may reflect different actuator bandwidths or accessible optical subspaces. Shared WFS integration and readout delays do not disappear from the absolute causal requirement in Eq.~\eqref{eq:causality}. They cancel only when comparing the \emph{differential} downstream delay of two actuators. A rapid electronic calculation or a fast switch transition is not an end-to-end measured latency.

Published HWO control concepts and space-coronagraph simulations provide concrete transfer-function models to test~\cite{Gerard2026,Potier2022}. Fast silicon-photonic switching has also been demonstrated in a different wavelength and device setting~\cite{Dupuis2015}. These works motivate measurement of the relevant transfer functions; they do not establish a visible/near-infrared HWO analyzer with the assumed loss, phase fidelity, radiation tolerance, and admitted duty. Photonic lantern and integrated-coronagraph architectures provide additional possible implementations~\cite{Xin2022,Sirbu2024}, each requiring its own calibrated optical map.

A scalar post-control wavefront-error (WFE) power spectral density $S_{\rm post}(f)$ has units $\mathrm{m}^2\,\mathrm{Hz}^{-1}$ when wavefront error is an optical path length in metres. Frequency $f$ is in hertz ($\mathrm{Hz}=\mathrm{s}^{-1}$); $f_{\rm Nyq}$ is the Nyquist frequency, half the wavefront sampling rate. For a frequency mask $\mathcal M(f)$ identifying bins that fail a declared control-rejection criterion but pass prediction and analyzer-response criteria, one may report
\begin{equation}
 f_{\rm WFE}=\frac{\int_0^{f_{\rm Nyq}}\mathcal M(f)S_{\rm post}(f)\,df}
                  {\int_0^{f_{\rm Nyq}}S_{\rm post}(f)\,df}.
 \label{eq:psd_mask}
\end{equation}
This is a variance fraction, not its square root. Even a nonzero value does not determine the leakage photon rate, complex-field cross spectra, evolution of the projector $S$, or the information-weighted eligible fraction. Without those quantities, integrating a convenient frequency window is not an eligibility demonstration. 

\subsection{Persistent latent rates and calibration error}\label{sec:latent}
When the conditional rate is uncertain over many photons, the count likelihood must retain that persistence. In a general block $a$ and interval $k$, with latent controlled state $x_k$, transfer matrix $H_{ak}$, and implemented effect $\Pi_{mak}(y_k)$, write
\begin{align}
 n_{mak}\mid x_k,y_k&\sim\operatorname{Pois}[\Lambda_{mak}(x_k)],\nonumber\\
 \Lambda_{mak}(x_k)&=\Delta t_k\bigl\{d_{ma}\nonumber\\
 &\quad+\Tr[\Pi_{mak}(y_k)H_{ak}R_{ak}(x_k)H_{ak}^{\dagger}]\bigr\}.
 \label{eq:full_likelihood}
\end{align}
Here $\Delta t_k$ is in seconds, $H_{ak}$ is a dimensionless complex-amplitude transfer matrix, and $R_{ak}$ and $d_{ma}$ are optical and detector-event rates in $\mathrm{s}^{-1}$. The count mean $\Lambda_{mak}$ is dimensionless. Marginalizing a persistent random rate produces a Cox process, namely a Poisson process conditional on a random intensity~\cite{Cox1955}, with
\begin{equation}
 \operatorname{Var}(n\mid y)=\E[\Lambda\mid y]
                      +\operatorname{Var}(\Lambda\mid y).
 \label{eq:cox}
\end{equation}
Shared latent states also create cross-cell covariances. Replacing $\Lambda$ by its mean loses this term. The data-processing identity that a marginalized score is the conditional expectation of the complete-data score gives $J_{\rm marginal}\le\E[J_{\rm complete}]$ when the latent law is parameter independent. It does not justify evaluating a Poisson information at the mean rate.

An analytically tractable nuisance illustrates why a ratio alone is insufficient. With no detector events, let a block-common log-amplitude $\alpha$ multiply the entire optical background, $\lambda_m=\tau_j[e^{\alpha}b_m+cp_m]$. At $\alpha=c=0$, completeness gives $\sum_mp_m=1$ and $\sum_mb_m=\beta\mathcal T$, so profiling this unknown amplitude subtracts
\begin{equation}
 J_{c\mid\alpha}^{(j)}
 =J_c^{(j)}-\frac{\tau_j}{\beta\mathcal T},
 \qquad \tau_{\rm f}=1,\quad\tau_{\rm c}=\tau.
 \label{eq:common_nuisance}
\end{equation}
At $r=10$, $\gamma=0$, and $\tau=0.8$, the resulting ratio is $1.7173$, larger than the Poisson ratio $1.3985$, although \emph{both} absolute informations have decreased. The common amplitude removes a total-count score component, leaving a larger relative benefit in conditioned channel contrast. This is a local nuisance calculation with a known multiplicative form, not a claim of robustness to arbitrary speckles. Nuisance-aware quantum estimation similarly requires the score directions and attainable measurement to be treated together~\cite{Tsang2020}.

Calibration error can be bounded without assigning an average ``crosstalk penalty.'' For a declared subset $\mathcal G$ of predicted states, suppose every retained informative output obeys
\begin{equation}
 |\widetilde g_m|\ge a_{\rm acc}|g_m|,\qquad
 \widetilde b_m\le b_m(1+\delta_b),
 \label{eq:calibration_bounds}
\end{equation}
relative to the already loss- and detector-charged modeled rates and slopes. A termwise comparison gives an implemented information at least $a_{\rm acc}^2/(1+\delta_b)$ times the modeled information on $\mathcal G$. The retained fraction must be weighted by information, $f_J=\E[\mathbf1_{\mathcal G}J]/\E[J]$, rather than by occurrence probability. If the modeled comparison has gain $G_{\rm model}$ against an unchanged calibrated fixed reference, a sufficient bound is
\begin{equation}
 G_{\rm impl}\ge\frac{f_Ja_{\rm acc}^2}{1+\delta_b}G_{\rm model}.
 \label{eq:calibration_gain}
\end{equation}
These are pointwise hypotheses to be established from transfer matrices. The acceptance $a_{\rm acc}$ is a rate-derivative ratio, not an intensity throughput, and effects already charged in the model must not be charged twice. The practical sensitivity of spatial-mode detection to imperfections is well documented~\cite{Linowski2025}.

\subsection{What would constitute an HWO result}
The theory supplies a measurement optimum, an implementation-aware exclusion, and conditional acceptance rules. Turning those into an HWO capability requires a common-controller field calculation with a named pupil, coronagraph, finite star, target scene, and disturbance model. The same withheld field histories must feed both receivers. That calculation must report the joint modal-rate and timing distribution. The direction-independent ceiling remains available whenever its optical and likelihood assumptions hold, but exact non-Haar gains require reoptimization of the fixed comparator; nonaligned floors require their full rate operators.

The next evidentiary step is a dynamic-leakage experiment with measured wavelength- and polarization-resolved transfer matrices, detector-event statistics, latency, and duty. The implemented likelihood must recover calibrated false alarms, interval coverage, and feature precision under held-out prediction errors. Independent channel gates must be demonstrated before using the independent-fallback guarantee. Finally, measured target-dependent exposure distributions must be propagated through atmospheric inference and scheduling; a favorable single-block ratio is not a universal multiplier of observing yield.

\section{Numerical verification and reproducibility}\label{sec:numerics}
All physical input rates, templates, duty conventions, and integration objectives used in the figures and tables are specified in the preceding sections. The accompanying executable calculation regenerates the numerical arrays. Bracketed roots are solved in dimensionless leakage or contrast ratios rather than directly on $10^{-9}$-scale contrasts. The isotropic integrals use 128-point Gauss--Legendre quadrature, with order doubling and independent adaptive integration as checks. Equation~\eqref{eq:closed_integral} is used only where its cancellation is controlled, and Eq.~\eqref{eq:limits} supplies the exact endpoint checks.

Independent matrix tests compare the explicit SLD in Eq.~\eqref{eq:sld_explicit}, a Sylvester solve, and the classical information of its eigenbasis for 10,000 randomly generated positive two-mode backgrounds. Random complete POVMs with 2, 3, 4, 6, and 8 outcomes provide implementation checks of Theorem~\ref{thm:fixed}, including full-rank effects and unequal aligned floors. These searches are diagnostics, not the proof. The Haar contrast bounds are checked using both the physical-rate expression and Eq.~\eqref{eq:exact_gain}. The direction-independent ceiling is additionally checked on 20,000 random discrete direction distributions with 1--32 support points, positive input rates from $10^{-6}$ to $1\,\mathrm{s}^{-1}$, and zero-floor cases. These tests compare the conditioned optimum to the feasible planet projection, so no unproved non-Haar fixed optimum is assumed. Independent matrix and output-rate checks cover 1,000 of these distributions. Independent phase quadrature at 200 effect overlaps checks Eq.~\eqref{eq:equatorial_fixed}; a known-static-direction control verifies that a fixed SLD instead removes the preloss advantage. Uniform rescaling of all physical rates by a factor $k$ verifies $J\mapsto J/k$ and gain invariance. The temporal integrals and their refinement are specified in Sec.~\ref{sec:temporal}; the reported crossings are stable well beyond the displayed digits, without elevating local optimization to global optimality.

For spectroscopy, 2,000 random identifiable template comparisons with two to eight equal-contrast blocks verify propagation of the $0.81$ ceiling through continuum profiling. All eight three-channel routes are enumerated to check the independent-gate corollary. Direct cell-rate Fisher matrices at 32 and 64 orientation nodes agree with Eq.~\eqref{eq:fisher_spectrum}. The profile solver in Eq.~\eqref{eq:profile_statistic} is independently checked against bounded scalar likelihood optimization. The finite-count streams in Table~\ref{tab:finite} contain 600,000 realizations in total and require at most five Newton steps in the executed calculation. Their seed, role separation, likelihood, and nuisance domain are specified in Sec.~\ref{sec:spectroscopy}.

The local limit is checked independently using the Poisson relative entropy. For null rates $b_m$ and slopes $g_m$,
\begin{equation}
 \frac{D(P_c\Vert P_0)}{t}
 =\sum_m\left[(b_m+cg_m)\ln\!\left(1+\frac{cg_m}{b_m}\right)-cg_m\right].
 \label{eq:kl}
\end{equation}
Expansion gives $2D/(tc^2)\to\sum_mg_m^2/b_m$. This checks the relation between the count likelihood and the local information, not equality of local and finite-alternative gains. Figures are generated from the same full-precision rates and roots as the tables. Decimal precision in illustrative astrophysical inputs is a reproducibility convention, not a claim about their physical accuracy.

\section{Conclusions}\label{sec:conclusions}
The useful resource is predictable coherence between the two selected spatial modes, not suppression of an estimated stellar mean. A conditioned analyzer can rotate this coherence into informative count differences, while its scalar transmission reduces both planet response and optical background. In the Haar ensemble, the exact fixed optimum is planet-matched single-mode injection, including known unequal aligned optical floors and matched detector-event rates. The two-output SLD attains the conditioned optimum in the stated single-photon class. The comparison therefore concerns measurement selection relative to a strong optical reference, not intrinsically quantum versus classical hardware.

The photons that share the planet mode place a stringent limit on this resource. For an isotropic optical floor, equal predictable-leakage and planet contrasts give $\Gwall\le0.8100$ for \emph{any distribution of recorded directions} at $\tau=0.80$ and $\drel=0.90$; Haar averaging tightens this to $0.7791$. A gain requires $x>x_{\rm nec}$, with $x_{\rm nec}\simeq2.248$ for arbitrary directions and $3.261$ for Haar directions; finite backgrounds impose stronger restrictions. In the explicit one-exozodi Haar budget, a 20\% exposure reduction requires $\Cpred\simeq3.41\times10^{-9}$, about 34 times the assigned planet contrast $10^{-10}$. These dimensionless contrasts refer to the selected modal pair, not a spatially averaged dark hole.

Predictability and directional diversity must coexist. In the specified von Mises--Fisher concentration model at $(r,\gamma,\tau,\kappa,\drel)=(10,0.5,0.8,0.1,0.9)$, the analytic certificate guarantees gain for $t_{\rm coh}/T_{\rm cal}<0.26598$ and at least 5\% local exposure saving for $t_{\rm coh}/T_{\rm cal}\le0.19514$. The more favorable final row of Table~\ref{tab:penalties} certifies 20\% saving up to $0.32729$. These are dimensionless diversity windows at the stated rates and prediction error, not measured instrumental timescales. The forecast must also meet the held-out error and admitted-duty requirements over the full physical horizon $t_{\rm ff}+T_{\rm hold}$, in seconds.

Spectral inference adds a further physical constraint: continuum photons in the reference bands determine the precision of the continuum beneath an absorption feature. With known information rates and costs, independent gates permit fixed-basis fallback that cannot decrease local continuum-profiled information; the three-channel example gives a gain of $1.0665$, or $6.24\%$ shorter asymptotic exposure. A shared gate requires a different global comparison. Neither result turns a single-block information ratio into a universal atmospheric-retrieval or observing-yield multiplier.

The decisive next step is a common-controller optical field calculation yielding the joint distribution of $r$, prediction-error power fraction $\epsilon$, and directional coherence time $t_{\rm coh}$, together with modal contrasts, transmission, and duty, followed by a dynamic-leakage experiment with measured transfer matrices. The present article establishes an equal-modal-contrast exclusion independent of direction statistics within the two-mode, conditional-Poisson, scalar-loss model,  and identifies conditional positive regimes that require joint physical validation. It supplies a quantitative test of when prior wavefront information is worth using to select the photon measurement, rather than only to interpret its outcomes.

\begin{acknowledgments}
The work described here was carried out at the Jet Propulsion Laboratory, California Institute of Technology, Pasadena, California, under a contract with the National Aeronautics and Space Administration. \textcopyright\ 2026. California Institute of Technology. Government sponsorship acknowledged.
\end{acknowledgments}


\begin{thebibliography}{30}%
\makeatletter
\providecommand \@ifxundefined [1]{%
 \@ifx{#1\undefined}
}%
\providecommand \@ifnum [1]{%
 \ifnum #1\expandafter \@firstoftwo
 \else \expandafter \@secondoftwo
 \fi
}%
\providecommand \@ifx [1]{%
 \ifx #1\expandafter \@firstoftwo
 \else \expandafter \@secondoftwo
 \fi
}%
\providecommand \natexlab [1]{#1}%
\providecommand \enquote  [1]{``#1''}%
\providecommand \bibnamefont  [1]{#1}%
\providecommand \bibfnamefont [1]{#1}%
\providecommand \citenamefont [1]{#1}%
\providecommand \href@noop [0]{\@secondoftwo}%
\providecommand \href [0]{\begingroup \@sanitize@url \@href}%
\providecommand \@href[1]{\@@startlink{#1}\@@href}%
\providecommand \@@href[1]{\endgroup#1\@@endlink}%
\providecommand \@sanitize@url [0]{\catcode `\\12\catcode `\$12\catcode
  `\&12\catcode `\#12\catcode `\^12\catcode `\_12\catcode `\%12\relax}%
\providecommand \@@startlink[1]{}%
\providecommand \@@endlink[0]{}%
\providecommand \url  [0]{\begingroup\@sanitize@url \@url }%
\providecommand \@url [1]{\endgroup\@href {#1}{\urlprefix }}%
\providecommand \urlprefix  [0]{URL }%
\providecommand \Eprint [0]{\href }%
\providecommand \doibase [0]{https://doi.org/}%
\providecommand \selectlanguage [0]{\@gobble}%
\providecommand \bibinfo  [0]{\@secondoftwo}%
\providecommand \bibfield  [0]{\@secondoftwo}%
\providecommand \translation [1]{[#1]}%
\providecommand \BibitemOpen [0]{}%
\providecommand \bibitemStop [0]{}%
\providecommand \bibitemNoStop [0]{.\EOS\space}%
\providecommand \EOS [0]{\spacefactor3000\relax}%
\providecommand \BibitemShut  [1]{\csname bibitem#1\endcsname}%
\let\auto@bib@innerbib\@empty
\bibitem [{\citenamefont {{National Academies of Sciences, Engineering, and
  Medicine}}(2021)}]{astro2020}%
  \BibitemOpen
  \bibfield  {author} {\bibinfo {author} {\bibnamefont {{National Academies of
  Sciences, Engineering, and Medicine}}},\ }\href
  {https://doi.org/10.17226/26141} {\emph {\bibinfo {title} {Pathways to
  Discovery in Astronomy and Astrophysics for the 2020s}}}\ (\bibinfo
  {publisher} {National Academies Press},\ \bibinfo {address} {Washington,
  DC},\ \bibinfo {year} {2021})\BibitemShut {NoStop}%
\bibitem [{\citenamefont {Stark}\ \emph {et~al.}(2024)\citenamefont {Stark},
  \citenamefont {Mennesson}, \citenamefont {Bryson}, \citenamefont {Ford},
  \citenamefont {Robinson}, \citenamefont {Belikov}, \citenamefont {Bolcar},
  \citenamefont {Feinberg}, \citenamefont {Guyon}, \citenamefont {Latouf},
  \citenamefont {Mandell}, \citenamefont {Rauscher}, \citenamefont {Sirbu},\
  and\ \citenamefont {Tuchow}}]{Stark2024}%
  \BibitemOpen
  \bibfield  {author} {\bibinfo {author} {\bibfnamefont {C.~C.}\ \bibnamefont
  {Stark}}, \bibinfo {author} {\bibfnamefont {B.}~\bibnamefont {Mennesson}},
  \bibinfo {author} {\bibfnamefont {S.~T.}\ \bibnamefont {Bryson}}, \bibinfo
  {author} {\bibfnamefont {E.~B.}\ \bibnamefont {Ford}}, \bibinfo {author}
  {\bibfnamefont {T.~D.}\ \bibnamefont {Robinson}}, \bibinfo {author}
  {\bibfnamefont {R.}~\bibnamefont {Belikov}}, \bibinfo {author} {\bibfnamefont
  {M.~R.}\ \bibnamefont {Bolcar}}, \bibinfo {author} {\bibfnamefont {L.~D.}\
  \bibnamefont {Feinberg}}, \bibinfo {author} {\bibfnamefont {O.}~\bibnamefont
  {Guyon}}, \bibinfo {author} {\bibfnamefont {N.}~\bibnamefont {Latouf}},
  \bibinfo {author} {\bibfnamefont {A.~M.}\ \bibnamefont {Mandell}}, \bibinfo
  {author} {\bibfnamefont {B.~J.}\ \bibnamefont {Rauscher}}, \bibinfo {author}
  {\bibfnamefont {D.}~\bibnamefont {Sirbu}},\ and\ \bibinfo {author}
  {\bibfnamefont {N.~W.}\ \bibnamefont {Tuchow}},\ }\bibfield  {title}
  {\bibinfo {title} {Paths to robust exoplanet science yield margin for the
  {Habitable Worlds Observatory}},\ }\href
  {https://doi.org/10.1117/1.JATIS.10.3.034006} {\bibfield  {journal} {\bibinfo
   {journal} {JATIS}\ }\textbf {\bibinfo {volume} {10}},\ \bibinfo {pages}
  {034006} (\bibinfo {year} {2024})}\BibitemShut {NoStop}%
\bibitem [{\citenamefont {Stark}\ \emph {et~al.}(2025)\citenamefont {Stark},
  \citenamefont {Steiger}, \citenamefont {Tokadjian}, \citenamefont
  {Savransky}, \citenamefont {Belikov}, \citenamefont {Chen}, \citenamefont
  {Krist}, \citenamefont {Macintosh}, \citenamefont {Morgan}, \citenamefont
  {Pueyo}, \citenamefont {Sirbu},\ and\ \citenamefont
  {Stapelfeldt}}]{Stark2025ETC}%
  \BibitemOpen
  \bibfield  {author} {\bibinfo {author} {\bibfnamefont {C.~C.}\ \bibnamefont
  {Stark}}, \bibinfo {author} {\bibfnamefont {S.}~\bibnamefont {Steiger}},
  \bibinfo {author} {\bibfnamefont {A.}~\bibnamefont {Tokadjian}}, \bibinfo
  {author} {\bibfnamefont {D.}~\bibnamefont {Savransky}}, \bibinfo {author}
  {\bibfnamefont {R.}~\bibnamefont {Belikov}}, \bibinfo {author} {\bibfnamefont
  {P.}~\bibnamefont {Chen}}, \bibinfo {author} {\bibfnamefont {J.}~\bibnamefont
  {Krist}}, \bibinfo {author} {\bibfnamefont {B.}~\bibnamefont {Macintosh}},
  \bibinfo {author} {\bibfnamefont {R.}~\bibnamefont {Morgan}}, \bibinfo
  {author} {\bibfnamefont {L.}~\bibnamefont {Pueyo}}, \bibinfo {author}
  {\bibfnamefont {D.}~\bibnamefont {Sirbu}},\ and\ \bibinfo {author}
  {\bibfnamefont {K.}~\bibnamefont {Stapelfeldt}},\ }\href@noop {} {\bibinfo
  {title} {Cross-model validation of coronagraphic exposure time calculators
  for the {Habitable Worlds Observatory}: A report from the exoplanet science
  yield sub-working group}} (\bibinfo {year} {2025}),\ \Eprint
  {https://arxiv.org/abs/2502.18556} {arXiv:2502.18556 [astro-ph.IM]}
  \BibitemShut {NoStop}%
\bibitem [{\citenamefont {Ruffio}\ \emph {et~al.}(2026)\citenamefont {Ruffio},
  \citenamefont {Steiger}, \citenamefont {Spohn}, \citenamefont {Macintosh},
  \citenamefont {Mawet}, \citenamefont {Pueyo}, \citenamefont {Mennesson},
  \citenamefont {Dacus}, \citenamefont {Wolff}, \citenamefont {Robinson},
  \citenamefont {Hu}, \citenamefont {Hoch}, \citenamefont {Konopacky},
  \citenamefont {Perrin}, \citenamefont {Savransky}, \citenamefont {McElwain},
  \citenamefont {Wright}, \citenamefont {Wang},\ and\ \citenamefont
  {Chen}}]{Ruffio2026}%
  \BibitemOpen
  \bibfield  {author} {\bibinfo {author} {\bibfnamefont {J.-B.}\ \bibnamefont
  {Ruffio}}, \bibinfo {author} {\bibfnamefont {S.}~\bibnamefont {Steiger}},
  \bibinfo {author} {\bibfnamefont {C.}~\bibnamefont {Spohn}}, \bibinfo
  {author} {\bibfnamefont {B.}~\bibnamefont {Macintosh}}, \bibinfo {author}
  {\bibfnamefont {D.}~\bibnamefont {Mawet}}, \bibinfo {author} {\bibfnamefont
  {L.}~\bibnamefont {Pueyo}}, \bibinfo {author} {\bibfnamefont
  {B.}~\bibnamefont {Mennesson}}, \bibinfo {author} {\bibfnamefont
  {B.}~\bibnamefont {Dacus}}, \bibinfo {author} {\bibfnamefont
  {N.}~\bibnamefont {Wolff}}, \bibinfo {author} {\bibfnamefont {T.~D.}\
  \bibnamefont {Robinson}}, \bibinfo {author} {\bibfnamefont {R.}~\bibnamefont
  {Hu}}, \bibinfo {author} {\bibfnamefont {K.}~\bibnamefont {Hoch}}, \bibinfo
  {author} {\bibfnamefont {Q.~M.}\ \bibnamefont {Konopacky}}, \bibinfo {author}
  {\bibfnamefont {M.~D.}\ \bibnamefont {Perrin}}, \bibinfo {author}
  {\bibfnamefont {D.}~\bibnamefont {Savransky}}, \bibinfo {author}
  {\bibfnamefont {M.~W.}\ \bibnamefont {McElwain}}, \bibinfo {author}
  {\bibfnamefont {S.~A.}\ \bibnamefont {Wright}}, \bibinfo {author}
  {\bibfnamefont {J.}~\bibnamefont {Wang}},\ and\ \bibinfo {author}
  {\bibfnamefont {P.}~\bibnamefont {Chen}},\ }\bibfield  {title} {\bibinfo
  {title} {Characterizing {Earth} analogs may require a moderate- or
  high-resolution spectrograph},\ }\href
  {https://doi.org/10.1117/1.JATIS.12.4.041020} {\bibfield  {journal} {\bibinfo
   {journal} {JATIS}\ }\textbf {\bibinfo {volume} {12}},\ \bibinfo {pages}
  {041020} (\bibinfo {year} {2026})}\BibitemShut {NoStop}%
\bibitem [{\citenamefont {Tesch}\ \emph {et~al.}(2026)\citenamefont {Tesch},
  \citenamefont {Shi}, \citenamefont {Redding}, \citenamefont {Lou},\ and\
  \citenamefont {Jewell}}]{Tesch2026}%
  \BibitemOpen
  \bibfield  {author} {\bibinfo {author} {\bibfnamefont {J.}~\bibnamefont
  {Tesch}}, \bibinfo {author} {\bibfnamefont {F.}~\bibnamefont {Shi}}, \bibinfo
  {author} {\bibfnamefont {D.}~\bibnamefont {Redding}}, \bibinfo {author}
  {\bibfnamefont {J.}~\bibnamefont {Lou}},\ and\ \bibinfo {author}
  {\bibfnamefont {J.}~\bibnamefont {Jewell}},\ }\bibfield  {title} {\bibinfo
  {title} {{Habitable Worlds Observatory} wavefront sensing and control:
  methods for ultrastability},\ }\href
  {https://doi.org/10.1117/1.JATIS.12.4.041012} {\bibfield  {journal} {\bibinfo
   {journal} {JATIS}\ }\textbf {\bibinfo {volume} {12}},\ \bibinfo {pages}
  {041012} (\bibinfo {year} {2026})}\BibitemShut {NoStop}%
\bibitem [{\citenamefont {Pogorelyuk}\ \emph {et~al.}(2021)\citenamefont
  {Pogorelyuk}, \citenamefont {Pueyo}, \citenamefont {Males}, \citenamefont
  {Cahoy},\ and\ \citenamefont {Kasdin}}]{Pogorelyuk2021}%
  \BibitemOpen
  \bibfield  {author} {\bibinfo {author} {\bibfnamefont {L.}~\bibnamefont
  {Pogorelyuk}}, \bibinfo {author} {\bibfnamefont {L.}~\bibnamefont {Pueyo}},
  \bibinfo {author} {\bibfnamefont {J.~R.}\ \bibnamefont {Males}}, \bibinfo
  {author} {\bibfnamefont {K.}~\bibnamefont {Cahoy}},\ and\ \bibinfo {author}
  {\bibfnamefont {N.~J.}\ \bibnamefont {Kasdin}},\ }\bibfield  {title}
  {\bibinfo {title} {Information-theoretical limits of recursive estimation and
  closed-loop control in high-contrast imaging},\ }\href
  {https://doi.org/10.3847/1538-4365/ac126d} {\bibfield  {journal} {\bibinfo
  {journal} {Astrophys. J. Suppl. Ser.}\ }\textbf {\bibinfo {volume} {256}},\
  \bibinfo {pages} {39} (\bibinfo {year} {2021})}\BibitemShut {NoStop}%
\bibitem [{\citenamefont {Potier}\ \emph {et~al.}(2022)\citenamefont {Potier},
  \citenamefont {Ruane}, \citenamefont {Stark}, \citenamefont {Chen},
  \citenamefont {Chopra}, \citenamefont {Dewell}, \citenamefont
  {Juanola-Parramon}, \citenamefont {Nordt}, \citenamefont {Pueyo},
  \citenamefont {Redding}, \citenamefont {Riggs},\ and\ \citenamefont
  {Sirbu}}]{Potier2022}%
  \BibitemOpen
  \bibfield  {author} {\bibinfo {author} {\bibfnamefont {A.}~\bibnamefont
  {Potier}}, \bibinfo {author} {\bibfnamefont {G.}~\bibnamefont {Ruane}},
  \bibinfo {author} {\bibfnamefont {C.~C.}\ \bibnamefont {Stark}}, \bibinfo
  {author} {\bibfnamefont {P.}~\bibnamefont {Chen}}, \bibinfo {author}
  {\bibfnamefont {A.}~\bibnamefont {Chopra}}, \bibinfo {author} {\bibfnamefont
  {L.~D.}\ \bibnamefont {Dewell}}, \bibinfo {author} {\bibfnamefont
  {R.}~\bibnamefont {Juanola-Parramon}}, \bibinfo {author} {\bibfnamefont
  {A.~A.}\ \bibnamefont {Nordt}}, \bibinfo {author} {\bibfnamefont {L.~A.}\
  \bibnamefont {Pueyo}}, \bibinfo {author} {\bibfnamefont {D.~C.}\ \bibnamefont
  {Redding}}, \bibinfo {author} {\bibfnamefont {A.~J.~E.}\ \bibnamefont
  {Riggs}},\ and\ \bibinfo {author} {\bibfnamefont {D.}~\bibnamefont {Sirbu}},\
  }\bibfield  {title} {\bibinfo {title} {Adaptive optics performance of a
  simulated coronagraph instrument on a large, segmented space telescope in
  steady state},\ }\href {https://doi.org/10.1117/1.JATIS.8.3.035002}
  {\bibfield  {journal} {\bibinfo  {journal} {JATIS}\ }\textbf {\bibinfo
  {volume} {8}},\ \bibinfo {pages} {035002} (\bibinfo {year}
  {2022})}\BibitemShut {NoStop}%
\bibitem [{\citenamefont {Tsang}\ \emph {et~al.}(2016)\citenamefont {Tsang},
  \citenamefont {Nair},\ and\ \citenamefont {Lu}}]{Tsang2016}%
  \BibitemOpen
  \bibfield  {author} {\bibinfo {author} {\bibfnamefont {M.}~\bibnamefont
  {Tsang}}, \bibinfo {author} {\bibfnamefont {R.}~\bibnamefont {Nair}},\ and\
  \bibinfo {author} {\bibfnamefont {X.-M.}\ \bibnamefont {Lu}},\ }\bibfield
  {title} {\bibinfo {title} {Quantum theory of superresolution for two
  incoherent optical point sources},\ }\href
  {https://doi.org/10.1103/PhysRevX.6.031033} {\bibfield  {journal} {\bibinfo
  {journal} {Phys. Rev. X}\ }\textbf {\bibinfo {volume} {6}},\ \bibinfo {pages}
  {031033} (\bibinfo {year} {2016})}\BibitemShut {NoStop}%
\bibitem [{\citenamefont {Deshler}\ \emph {et~al.}(2025)\citenamefont
  {Deshler}, \citenamefont {Ozer}, \citenamefont {Ashok},\ and\ \citenamefont
  {Guha}}]{Deshler2025}%
  \BibitemOpen
  \bibfield  {author} {\bibinfo {author} {\bibfnamefont {N.}~\bibnamefont
  {Deshler}}, \bibinfo {author} {\bibfnamefont {I.}~\bibnamefont {Ozer}},
  \bibinfo {author} {\bibfnamefont {A.}~\bibnamefont {Ashok}},\ and\ \bibinfo
  {author} {\bibfnamefont {S.}~\bibnamefont {Guha}},\ }\bibfield  {title}
  {\bibinfo {title} {Experimental demonstration of a quantum-optimal
  coronagraph using spatial mode sorters},\ }\href
  {https://doi.org/10.1364/OPTICA.545414} {\bibfield  {journal} {\bibinfo
  {journal} {Optica}\ }\textbf {\bibinfo {volume} {12}},\ \bibinfo {pages}
  {518} (\bibinfo {year} {2025})}\BibitemShut {NoStop}%
\bibitem [{\citenamefont {Deshler}\ \emph {et~al.}(2026)\citenamefont
  {Deshler}, \citenamefont {Haffert},\ and\ \citenamefont
  {Ashok}}]{Deshler2026}%
  \BibitemOpen
  \bibfield  {author} {\bibinfo {author} {\bibfnamefont {N.}~\bibnamefont
  {Deshler}}, \bibinfo {author} {\bibfnamefont {S.}~\bibnamefont {Haffert}},\
  and\ \bibinfo {author} {\bibfnamefont {A.}~\bibnamefont {Ashok}},\ }\bibfield
   {title} {\bibinfo {title} {Quantum limits of exoplanet detection and
  localization},\ }\href {https://doi.org/10.1103/g6kr-mqdj} {\bibfield
  {journal} {\bibinfo  {journal} {Phys. Rev. A}\ }\textbf {\bibinfo {volume}
  {113}},\ \bibinfo {pages} {042406} (\bibinfo {year} {2026})}\BibitemShut
  {NoStop}%
\bibitem [{\citenamefont {Xin}\ \emph {et~al.}(2026{\natexlab{a}})\citenamefont
  {Xin}, \citenamefont {Haffert} \emph {et~al.}}]{Xin2026}%
  \BibitemOpen
  \bibfield  {author} {\bibinfo {author} {\bibfnamefont {Y.}~\bibnamefont
  {Xin}}, \bibinfo {author} {\bibfnamefont {S.}~\bibnamefont {Haffert}}, \emph
  {et~al.},\ }\href@noop {} {\bibinfo {title} {Quantum-optimal coronagraphy
  with spatial mode sorting for direct exoplanet observations}} (\bibinfo
  {year} {2026}{\natexlab{a}}),\ \Eprint {https://arxiv.org/abs/2607.02065}
  {arXiv:2607.02065 [astro-ph.IM]} \BibitemShut {NoStop}%
\bibitem [{\citenamefont {Xin}\ \emph {et~al.}(2026{\natexlab{b}})\citenamefont
  {Xin}, \citenamefont {Haffert},\ and\ \citenamefont {Landman}}]{XinHWO2026}%
  \BibitemOpen
  \bibfield  {author} {\bibinfo {author} {\bibfnamefont {Y.}~\bibnamefont
  {Xin}}, \bibinfo {author} {\bibfnamefont {S.~Y.}\ \bibnamefont {Haffert}},\
  and\ \bibinfo {author} {\bibfnamefont {R.}~\bibnamefont {Landman}},\
  }\bibfield  {title} {\bibinfo {title} {Optimal mode-sorting coronagraphy:
  limits of single-moded measurements for the {Habitable Worlds Observatory}},\
  }in\ \href {https://doi.org/10.1117/12.3101438} {\emph {\bibinfo {booktitle}
  {Optical and Infrared Interferometry and Imaging {X}}}},\ \bibinfo {series}
  {Proc. SPIE}, Vol.\ \bibinfo {volume} {14148}\ (\bibinfo {year} {2026})\ p.\
  \bibinfo {pages} {141483I}\BibitemShut {NoStop}%
\bibitem [{\citenamefont {Haffert}\ \emph {et~al.}(2026)\citenamefont
  {Haffert}, \citenamefont {Xin},\ and\ \citenamefont {Landman}}]{Haffert2026}%
  \BibitemOpen
  \bibfield  {author} {\bibinfo {author} {\bibfnamefont {S.~Y.}\ \bibnamefont
  {Haffert}}, \bibinfo {author} {\bibfnamefont {Y.}~\bibnamefont {Xin}},\ and\
  \bibinfo {author} {\bibfnamefont {R.}~\bibnamefont {Landman}},\ }\bibfield
  {title} {\bibinfo {title} {Quantum-optimal instruments for high-contrast
  imaging with multi-plane light converters},\ }in\ \href
  {https://doi.org/10.1117/12.3104662} {\emph {\bibinfo {booktitle} {Advances
  in Optical and Mechanical Technologies for Telescopes and Instrumentation
  {VII}}}},\ \bibinfo {series} {Proc. SPIE}, Vol.\ \bibinfo {volume} {14154}\
  (\bibinfo {year} {2026})\ p.\ \bibinfo {pages} {1415419}\BibitemShut
  {NoStop}%
\bibitem [{\citenamefont {Frazin}(2013)}]{Frazin2013}%
  \BibitemOpen
  \bibfield  {author} {\bibinfo {author} {\bibfnamefont {R.~A.}\ \bibnamefont
  {Frazin}},\ }\bibfield  {title} {\bibinfo {title} {Utilization of the
  wavefront sensor and short-exposure images for simultaneous estimation of
  quasi-static aberration and exoplanet intensity},\ }\href
  {https://doi.org/10.1088/0004-637X/767/1/21} {\bibfield  {journal} {\bibinfo
  {journal} {ApJ}\ }\textbf {\bibinfo {volume} {767}},\ \bibinfo {pages} {21}
  (\bibinfo {year} {2013})}\BibitemShut {NoStop}%
\bibitem [{\citenamefont {Rodack}\ \emph {et~al.}(2021)\citenamefont {Rodack},
  \citenamefont {Frazin}, \citenamefont {Males},\ and\ \citenamefont
  {Guyon}}]{Rodack2021}%
  \BibitemOpen
  \bibfield  {author} {\bibinfo {author} {\bibfnamefont {A.~T.}\ \bibnamefont
  {Rodack}}, \bibinfo {author} {\bibfnamefont {R.~A.}\ \bibnamefont {Frazin}},
  \bibinfo {author} {\bibfnamefont {J.~R.}\ \bibnamefont {Males}},\ and\
  \bibinfo {author} {\bibfnamefont {O.}~\bibnamefont {Guyon}},\ }\bibfield
  {title} {\bibinfo {title} {Millisecond exoplanet imaging: {I}. method and
  simulation results},\ }\href {https://doi.org/10.1364/JOSAA.426046}
  {\bibfield  {journal} {\bibinfo  {journal} {J. Opt. Soc. Am. A}\ }\textbf
  {\bibinfo {volume} {38}},\ \bibinfo {pages} {1541} (\bibinfo {year}
  {2021})}\BibitemShut {NoStop}%
\bibitem [{\citenamefont {Xin}\ \emph {et~al.}(2024)\citenamefont {Xin},
  \citenamefont {Pueyo}, \citenamefont {Laugier}, \citenamefont {Pogorelyuk},
  \citenamefont {Douglas}, \citenamefont {Pope},\ and\ \citenamefont
  {Cahoy}}]{XinWFE2024}%
  \BibitemOpen
  \bibfield  {author} {\bibinfo {author} {\bibfnamefont {Y.}~\bibnamefont
  {Xin}}, \bibinfo {author} {\bibfnamefont {L.}~\bibnamefont {Pueyo}}, \bibinfo
  {author} {\bibfnamefont {R.}~\bibnamefont {Laugier}}, \bibinfo {author}
  {\bibfnamefont {L.}~\bibnamefont {Pogorelyuk}}, \bibinfo {author}
  {\bibfnamefont {E.~S.}\ \bibnamefont {Douglas}}, \bibinfo {author}
  {\bibfnamefont {B.~J.~S.}\ \bibnamefont {Pope}},\ and\ \bibinfo {author}
  {\bibfnamefont {K.~L.}\ \bibnamefont {Cahoy}},\ }\bibfield  {title} {\bibinfo
  {title} {Coronagraphic data post-processing using projections on instrumental
  modes},\ }\href {https://doi.org/10.3847/1538-4357/ad1879} {\bibfield
  {journal} {\bibinfo  {journal} {ApJ}\ }\textbf {\bibinfo {volume} {963}},\
  \bibinfo {pages} {96} (\bibinfo {year} {2024})}\BibitemShut {NoStop}%
\bibitem [{\citenamefont {Ballester}\ \emph {et~al.}(2008)\citenamefont
  {Ballester}, \citenamefont {Wehner},\ and\ \citenamefont
  {Winter}}]{Ballester2008}%
  \BibitemOpen
  \bibfield  {author} {\bibinfo {author} {\bibfnamefont {M.~A.}\ \bibnamefont
  {Ballester}}, \bibinfo {author} {\bibfnamefont {S.}~\bibnamefont {Wehner}},\
  and\ \bibinfo {author} {\bibfnamefont {A.}~\bibnamefont {Winter}},\
  }\bibfield  {title} {\bibinfo {title} {State discrimination with
  post-measurement information},\ }\href
  {https://doi.org/10.1109/TIT.2008.928276} {\bibfield  {journal} {\bibinfo
  {journal} {IEEE Trans. Inf. Theory}\ }\textbf {\bibinfo {volume} {54}},\
  \bibinfo {pages} {4183} (\bibinfo {year} {2008})}\BibitemShut {NoStop}%
\bibitem [{\citenamefont {Carmeli}\ \emph {et~al.}(2018)\citenamefont
  {Carmeli}, \citenamefont {Heinosaari},\ and\ \citenamefont
  {Toigo}}]{Carmeli2018}%
  \BibitemOpen
  \bibfield  {author} {\bibinfo {author} {\bibfnamefont {C.}~\bibnamefont
  {Carmeli}}, \bibinfo {author} {\bibfnamefont {T.}~\bibnamefont
  {Heinosaari}},\ and\ \bibinfo {author} {\bibfnamefont {A.}~\bibnamefont
  {Toigo}},\ }\bibfield  {title} {\bibinfo {title} {State discrimination with
  postmeasurement information and incompatibility of quantum measurements},\
  }\href {https://doi.org/10.1103/PhysRevA.98.012126} {\bibfield  {journal}
  {\bibinfo  {journal} {Phys. Rev. A}\ }\textbf {\bibinfo {volume} {98}},\
  \bibinfo {pages} {012126} (\bibinfo {year} {2018})}\BibitemShut {NoStop}%
\bibitem [{\citenamefont {Choi}\ \emph {et~al.}(2026)\citenamefont {Choi},
  \citenamefont {Baac}, \citenamefont {Jacob},\ and\ \citenamefont
  {Chung}}]{Choi2026}%
  \BibitemOpen
  \bibfield  {author} {\bibinfo {author} {\bibfnamefont {H.}~\bibnamefont
  {Choi}}, \bibinfo {author} {\bibfnamefont {H.~W.}\ \bibnamefont {Baac}},
  \bibinfo {author} {\bibfnamefont {Z.}~\bibnamefont {Jacob}},\ and\ \bibinfo
  {author} {\bibfnamefont {H.}~\bibnamefont {Chung}},\ }\href@noop {} {\bibinfo
  {title} {Exoplanet detection using adaptive quantum-optimal measurement}}
  (\bibinfo {year} {2026}),\ \Eprint {https://arxiv.org/abs/2607.06931v1}
  {arXiv:2607.06931v1 [physics.optics]} \BibitemShut {NoStop}%
\bibitem [{\citenamefont {Turyshev}(2026)}]{Turyshev2026Benchmarks}%
  \BibitemOpen
  \bibfield  {author} {\bibinfo {author} {\bibfnamefont {S.~G.}\ \bibnamefont
  {Turyshev}},\ }\href@noop {} {\bibinfo {title} {{Do Quantum Measurement
  Advantages Survive to Astrophysical Inference? Seven Benchmarks in Optical
  Interferometry and Imaging}}} (\bibinfo {year} {2026}),\ \Eprint
  {https://arxiv.org/abs/2609.11974} {arXiv:2609.11974 [physics.gen-ph]}
  \BibitemShut {NoStop}%
\bibitem [{\citenamefont {Braunstein}\ and\ \citenamefont
  {Caves}(1994)}]{BraunsteinCaves1994}%
  \BibitemOpen
  \bibfield  {author} {\bibinfo {author} {\bibfnamefont {S.~L.}\ \bibnamefont
  {Braunstein}}\ and\ \bibinfo {author} {\bibfnamefont {C.~M.}\ \bibnamefont
  {Caves}},\ }\bibfield  {title} {\bibinfo {title} {Statistical distance and
  the geometry of quantum states},\ }\href
  {https://doi.org/10.1103/PhysRevLett.72.3439} {\bibfield  {journal} {\bibinfo
   {journal} {PRL}\ }\textbf {\bibinfo {volume} {72}},\ \bibinfo {pages} {3439}
  (\bibinfo {year} {1994})}\BibitemShut {NoStop}%
\bibitem [{\citenamefont {Huang}\ \emph {et~al.}(2023)\citenamefont {Huang},
  \citenamefont {Schwab},\ and\ \citenamefont {Lupo}}]{Huang2023}%
  \BibitemOpen
  \bibfield  {author} {\bibinfo {author} {\bibfnamefont {Z.}~\bibnamefont
  {Huang}}, \bibinfo {author} {\bibfnamefont {C.}~\bibnamefont {Schwab}},\ and\
  \bibinfo {author} {\bibfnamefont {C.}~\bibnamefont {Lupo}},\ }\bibfield
  {title} {\bibinfo {title} {Ultimate limits of exoplanet spectroscopy: A
  quantum approach},\ }\href {https://doi.org/10.1103/PhysRevA.107.022409}
  {\bibfield  {journal} {\bibinfo  {journal} {Phys. Rev. A}\ }\textbf {\bibinfo
  {volume} {107}},\ \bibinfo {pages} {022409} (\bibinfo {year}
  {2023})}\BibitemShut {NoStop}%
\bibitem [{\citenamefont {Santamaria}\ \emph {et~al.}(2025)\citenamefont
  {Santamaria}, \citenamefont {Sgobba}, \citenamefont {Pallotti},\ and\
  \citenamefont {Lupo}}]{Santamaria2025}%
  \BibitemOpen
  \bibfield  {author} {\bibinfo {author} {\bibfnamefont {L.}~\bibnamefont
  {Santamaria}}, \bibinfo {author} {\bibfnamefont {F.}~\bibnamefont {Sgobba}},
  \bibinfo {author} {\bibfnamefont {D.}~\bibnamefont {Pallotti}},\ and\
  \bibinfo {author} {\bibfnamefont {C.}~\bibnamefont {Lupo}},\ }\bibfield
  {title} {\bibinfo {title} {Single-photon super-resolved spectroscopy from
  spatial-mode demultiplexing},\ }\href {https://doi.org/10.1364/PRJ.544197}
  {\bibfield  {journal} {\bibinfo  {journal} {Photonics Res.}\ }\textbf
  {\bibinfo {volume} {13}},\ \bibinfo {pages} {865} (\bibinfo {year}
  {2025})}\BibitemShut {NoStop}%
\bibitem [{\citenamefont {Gerard}\ \emph {et~al.}(2026)\citenamefont {Gerard},
  \citenamefont {Geringer-Sameth}, \citenamefont {Sengupta}, \citenamefont
  {Perloff}, \citenamefont {Sanchez}, \citenamefont {Waswa}, \citenamefont
  {Laguna}, \citenamefont {Jensen-Clem}, \citenamefont {Poyneer},\ and\
  \citenamefont {Eckart}}]{Gerard2026}%
  \BibitemOpen
  \bibfield  {author} {\bibinfo {author} {\bibfnamefont {B.~L.}\ \bibnamefont
  {Gerard}}, \bibinfo {author} {\bibfnamefont {A.}~\bibnamefont
  {Geringer-Sameth}}, \bibinfo {author} {\bibfnamefont {A.~R.}\ \bibnamefont
  {Sengupta}}, \bibinfo {author} {\bibfnamefont {A.}~\bibnamefont {Perloff}},
  \bibinfo {author} {\bibfnamefont {D.~F.}\ \bibnamefont {Sanchez}}, \bibinfo
  {author} {\bibfnamefont {P.}~\bibnamefont {Waswa}}, \bibinfo {author}
  {\bibfnamefont {C.}~\bibnamefont {Laguna}}, \bibinfo {author} {\bibfnamefont
  {R.}~\bibnamefont {Jensen-Clem}}, \bibinfo {author} {\bibfnamefont
  {L.}~\bibnamefont {Poyneer}},\ and\ \bibinfo {author} {\bibfnamefont
  {M.}~\bibnamefont {Eckart}},\ }\bibfield  {title} {\bibinfo {title}
  {{WaveDriver}: a laser guide star {AO} system for {HWO}},\ }\href
  {https://doi.org/10.1117/1.JATIS.12.4.041015} {\bibfield  {journal} {\bibinfo
   {journal} {J. Astron. Telesc. Instrum. Syst.}\ }\textbf {\bibinfo {volume}
  {12}},\ \bibinfo {pages} {041015} (\bibinfo {year} {2026})}\BibitemShut
  {NoStop}%
\bibitem [{\citenamefont {Dupuis}\ \emph {et~al.}(2015)\citenamefont {Dupuis},
  \citenamefont {Lee}, \citenamefont {Rylyakov}, \citenamefont {Kuchta},
  \citenamefont {Baks}, \citenamefont {Orcutt}, \citenamefont {Gill},
  \citenamefont {Green},\ and\ \citenamefont {Schow}}]{Dupuis2015}%
  \BibitemOpen
  \bibfield  {author} {\bibinfo {author} {\bibfnamefont {N.}~\bibnamefont
  {Dupuis}}, \bibinfo {author} {\bibfnamefont {B.~G.}\ \bibnamefont {Lee}},
  \bibinfo {author} {\bibfnamefont {A.~V.}\ \bibnamefont {Rylyakov}}, \bibinfo
  {author} {\bibfnamefont {D.~M.}\ \bibnamefont {Kuchta}}, \bibinfo {author}
  {\bibfnamefont {C.~W.}\ \bibnamefont {Baks}}, \bibinfo {author}
  {\bibfnamefont {J.~S.}\ \bibnamefont {Orcutt}}, \bibinfo {author}
  {\bibfnamefont {D.~M.}\ \bibnamefont {Gill}}, \bibinfo {author}
  {\bibfnamefont {W.~M.~J.}\ \bibnamefont {Green}},\ and\ \bibinfo {author}
  {\bibfnamefont {C.~L.}\ \bibnamefont {Schow}},\ }\bibfield  {title} {\bibinfo
  {title} {Design and fabrication of low-insertion-loss and low-crosstalk
  broadband {$2\times2$} {Mach--Zehnder} silicon photonic switches},\ }\href
  {https://doi.org/10.1109/JLT.2015.2446463} {\bibfield  {journal} {\bibinfo
  {journal} {J. Lightwave Technol.}\ }\textbf {\bibinfo {volume} {33}},\
  \bibinfo {pages} {3597} (\bibinfo {year} {2015})}\BibitemShut {NoStop}%
\bibitem [{\citenamefont {Xin}\ \emph {et~al.}(2022)\citenamefont {Xin},
  \citenamefont {Jovanovic}, \citenamefont {Ruane}, \citenamefont {Mawet},
  \citenamefont {Fitzgerald}, \citenamefont {Echeverri}, \citenamefont {Lin},
  \citenamefont {Leon-Saval}, \citenamefont {Gatkine}, \citenamefont {Kim},
  \citenamefont {Norris},\ and\ \citenamefont {Sallum}}]{Xin2022}%
  \BibitemOpen
  \bibfield  {author} {\bibinfo {author} {\bibfnamefont {Y.}~\bibnamefont
  {Xin}}, \bibinfo {author} {\bibfnamefont {N.}~\bibnamefont {Jovanovic}},
  \bibinfo {author} {\bibfnamefont {G.}~\bibnamefont {Ruane}}, \bibinfo
  {author} {\bibfnamefont {D.}~\bibnamefont {Mawet}}, \bibinfo {author}
  {\bibfnamefont {M.~P.}\ \bibnamefont {Fitzgerald}}, \bibinfo {author}
  {\bibfnamefont {D.}~\bibnamefont {Echeverri}}, \bibinfo {author}
  {\bibfnamefont {J.}~\bibnamefont {Lin}}, \bibinfo {author} {\bibfnamefont
  {S.~G.}\ \bibnamefont {Leon-Saval}}, \bibinfo {author} {\bibfnamefont
  {P.}~\bibnamefont {Gatkine}}, \bibinfo {author} {\bibfnamefont {Y.~J.}\
  \bibnamefont {Kim}}, \bibinfo {author} {\bibfnamefont {B.}~\bibnamefont
  {Norris}},\ and\ \bibinfo {author} {\bibfnamefont {S.}~\bibnamefont
  {Sallum}},\ }\bibfield  {title} {\bibinfo {title} {Efficient detection and
  characterization of exoplanets within the diffraction limit: Nulling with a
  mode-selective photonic lantern},\ }\href
  {https://doi.org/10.3847/1538-4357/ac9284} {\bibfield  {journal} {\bibinfo
  {journal} {ApJ}\ }\textbf {\bibinfo {volume} {938}},\ \bibinfo {pages} {140}
  (\bibinfo {year} {2022})}\BibitemShut {NoStop}%
\bibitem [{\citenamefont {Sirbu}\ \emph {et~al.}(2024)\citenamefont {Sirbu},
  \citenamefont {Belikov}, \citenamefont {Fogarty}, \citenamefont {Valdez},
  \citenamefont {Sun}, \citenamefont {Kroo}, \citenamefont {Solgaard},
  \citenamefont {Miller},\ and\ \citenamefont {Guyon}}]{Sirbu2024}%
  \BibitemOpen
  \bibfield  {author} {\bibinfo {author} {\bibfnamefont {D.}~\bibnamefont
  {Sirbu}}, \bibinfo {author} {\bibfnamefont {R.}~\bibnamefont {Belikov}},
  \bibinfo {author} {\bibfnamefont {K.}~\bibnamefont {Fogarty}}, \bibinfo
  {author} {\bibfnamefont {C.}~\bibnamefont {Valdez}}, \bibinfo {author}
  {\bibfnamefont {Z.}~\bibnamefont {Sun}}, \bibinfo {author} {\bibfnamefont
  {A.}~\bibnamefont {Kroo}}, \bibinfo {author} {\bibfnamefont {O.}~\bibnamefont
  {Solgaard}}, \bibinfo {author} {\bibfnamefont {D.~A.~B.}\ \bibnamefont
  {Miller}},\ and\ \bibinfo {author} {\bibfnamefont {O.}~\bibnamefont
  {Guyon}},\ }\bibfield  {title} {\bibinfo {title} {{AstroPIC}: near-infrared
  photonic integrated circuit coronagraph architecture for the {Habitable
  Worlds Observatory}},\ }in\ \href {https://doi.org/10.1117/12.3020518} {\emph
  {\bibinfo {booktitle} {Space Telescopes and Instrumentation 2024: Optical,
  Infrared, and Millimeter Wave}}},\ \bibinfo {series} {Proc. SPIE}, Vol.\
  \bibinfo {volume} {13092}\ (\bibinfo {year} {2024})\ p.\ \bibinfo {pages}
  {130921T}\BibitemShut {NoStop}%
\bibitem [{\citenamefont {Cox}(1955)}]{Cox1955}%
  \BibitemOpen
  \bibfield  {author} {\bibinfo {author} {\bibfnamefont {D.~R.}\ \bibnamefont
  {Cox}},\ }\bibfield  {title} {\bibinfo {title} {Some statistical methods
  connected with series of events},\ }\href
  {https://doi.org/10.1111/j.2517-6161.1955.tb00188.x} {\bibfield  {journal}
  {\bibinfo  {journal} {J. R. Stat. Soc. Ser. B}\ }\textbf {\bibinfo {volume}
  {17}},\ \bibinfo {pages} {129} (\bibinfo {year} {1955})}\BibitemShut
  {NoStop}%
\bibitem [{\citenamefont {Tsang}\ \emph {et~al.}(2020)\citenamefont {Tsang},
  \citenamefont {Albarelli},\ and\ \citenamefont {Datta}}]{Tsang2020}%
  \BibitemOpen
  \bibfield  {author} {\bibinfo {author} {\bibfnamefont {M.}~\bibnamefont
  {Tsang}}, \bibinfo {author} {\bibfnamefont {F.}~\bibnamefont {Albarelli}},\
  and\ \bibinfo {author} {\bibfnamefont {A.}~\bibnamefont {Datta}},\ }\bibfield
   {title} {\bibinfo {title} {Quantum semiparametric estimation},\ }\href
  {https://doi.org/10.1103/PhysRevX.10.031023} {\bibfield  {journal} {\bibinfo
  {journal} {Phys. Rev. X}\ }\textbf {\bibinfo {volume} {10}},\ \bibinfo
  {pages} {031023} (\bibinfo {year} {2020})}\BibitemShut {NoStop}%
\bibitem [{\citenamefont {Linowski}\ \emph {et~al.}(2025)\citenamefont
  {Linowski}, \citenamefont {Schlichtholz},\ and\ \citenamefont
  {Sorelli}}]{Linowski2025}%
  \BibitemOpen
  \bibfield  {author} {\bibinfo {author} {\bibfnamefont {T.}~\bibnamefont
  {Linowski}}, \bibinfo {author} {\bibfnamefont {K.}~\bibnamefont
  {Schlichtholz}},\ and\ \bibinfo {author} {\bibfnamefont {G.}~\bibnamefont
  {Sorelli}},\ }\bibfield  {title} {\bibinfo {title} {Quantum-inspired
  exoplanet detection in the presence of experimental imperfections},\ }\href
  {https://doi.org/10.1103/qq99-jmpv} {\bibfield  {journal} {\bibinfo
  {journal} {Phys. Rev. Appl.}\ }\textbf {\bibinfo {volume} {24}},\ \bibinfo
  {pages} {044052} (\bibinfo {year} {2025})}\BibitemShut {NoStop}%
\end{thebibliography}

%

\end{document}